\documentclass[aps,prx,twocolumn,superscriptaddress,nobibnotes]{revtex4-2}

\usepackage{amsthm, color, mathtools, amsmath, amssymb, mathtools, graphicx, float, verbatim, xcolor, placeins, needspace}

\usepackage{bm}
\usepackage{enumerate}
\usepackage{braket}
\usepackage{microtype} 
\usepackage[british]{babel}
\usepackage[unicode=true,bookmarks=true,bookmarksnumbered=false,bookmarksopen=false,breaklinks=false,pdfborder={0 0 1}, backref=false,colorlinks=true,linkcolor=blue,urlcolor=blue,citecolor=blue]{hyperref}

\newcommand{\todai}{Department of Physics, Graduate School of Science, The University of Tokyo, Hongo 7-3-1, Bunkyo-ku, Tokyo 113-0033, Japan}
\newcommand{\transscale}{Trans-Scale Quantum Science Institute, The University of Tokyo, Hongo 7-3-1, Bunkyo-ku, Tokyo 113-0033, Japan}

\newtheorem{lem}{Lemma}

\newtheorem{Theorem}{Theorem}
\newtheorem{Corollary}[Theorem]{Corollary}

\definecolor{darkgreen}{RGB}{0,200,107}

\definecolor{darkorange}{RGB}{255,100,0}

\usepackage{physics}
\usepackage[normalem]{ulem}

\newcommand{\RR}{\mathbb{R}}
\newcommand{\CC}{\mathbb{C}}
\newcommand{\SU}{\mathrm{SU}}
\newcommand{\su}{\mathfrak{su}}
\newcommand{\SWAP}{\mathrm{SWAP}}
\newcommand{\dbra}[1]{\langle\!\langle {#1} \vert}
\newcommand{\dket}[1]{\vert {#1} \rangle\!\rangle}
\NewDocumentCommand\dbraket{m m}{
    \IfNoValueTF{#1}{\langle\!\langle {#2} \vert {#2} \rangle\!\rangle}{\langle\!\langle {#1} \vert {#2} \rangle\!\rangle}
}
\newcommand{\dketbra}[1]{\vert {#1} \rangle\!\rangle\!\langle\!\langle {#1} \vert}

\newcommand{\mcA}{\mathcal{A}}
\newcommand{\mcB}{\mathcal{B}}

\newcommand{\mcE}{\mathcal{E}}

\newcommand{\mcH}{\mathcal{H}}
\newcommand{\mcI}{\mathcal{I}}

\newcommand{\mcL}{\mathcal{L}}

\begin{document}

\title{Analytical Lower Bound on Query Complexity for Transformations of Unknown Unitary Operations}
\begin{abstract}
    Recent developments have revealed deterministic and exact protocols for performing complex conjugation, inversion, and transposition of a general $d$-dimensional unknown unitary operation using a finite number of queries to a black-box unitary operation. In this work, we establish analytical lower bounds for the query complexity of unitary inversion, transposition, and complex conjugation, which hold even if the input unitary is an unknown logarithmic-depth unitary.
    Specifically, our lower bound of $d^2$ for unitary inversion demonstrates the asymptotic optimality of the deterministic exact inversion protocol, which operates with $O(d^2)$ queries. We introduce a novel framework utilizing differentiation to derive these lower bounds on query complexity for general differentiable functions $f: \SU(d)\to \SU(d)$. As a corollary, we prove that a catalytic protocol -- a new concept recently noted in the study of exact unitary inversion -- is impossible for unitary complex conjugation. Furthermore, we extend our framework to the partially known setting, where the input unitary operation is promised to be within a subgroup of $\SU(d)$ and the probabilistic setting, where transformations succeed probabilistically.
\end{abstract}

\author{Tatsuki Odake}
\affiliation{\todai}
\author{Satoshi Yoshida}
\email{satoshiyoshida.phys@gmail.com}
\affiliation{\todai}
\author{Mio Murao}
\email{murao@phys.s.u-tokyo.ac.jp}
\affiliation{\todai}
\affiliation{\transscale}

\date{\today}

\maketitle

{\it Introduction.}---
No-go theorems have played a vital role in the history of quantum information theory.
The no-cloning theorem prohibits cloning of an \textit{unknown} quantum state, and this property of quantum mechanics led to the invention of cryptographic primitives such as quantum key distribution \cite{wootters1982single, bennett2014quantum}.
Researchers have considered information processing tasks for unknown quantum states, such as broadcasting quantum information, and shown no-go theorems for these tasks.
These no-go theorems play a complementary role to the go results, which are probabilistic or approximate protocols to implement transformations of unknown quantum states \cite{werner1998optimal, agrawal2002probabilistic, buzek2000optimal}.
They provide an understanding of the nature of quantum states as an information carrier, which leads to ideas for implementing quantum protocols and establishes the foundation of quantum mechanics.

Recently, transformations of unknown \textit{unitary operations}, which implement an unitary operation $f(U)$ using multiple queries to an unknown unitary operation $U$ for a given function $f:\SU(d)\to\SU(d)$, have been extensively studied.
Such a transformation is relevant for remote quantum computing~\cite{yang2019optimal}, quantum control~\cite{navascues2018resetting}, learning of quantum dynamics~\cite{garttner2017measuring, li2017measuring} and quantum functional programming~\cite{bisio2019theoretical}.
One of the tasks in remote quantum computing is to perform $f(U)$ for a unitary operation $U$ held by Charlie on a quantum state $\ket{\psi}$ held by David~\cite{yang2019optimal}.
Alice has a blind access to a quantum state $\ket{\psi}$, Bob has a blind access to a unitary operation $U$, and Alice and Bob communicate quantum states to implement $f(U)\ket{\psi}$.
In the quantum control, we can cancel out an unwanted Hamiltonian dynamics given by unitary operation $U$ by implementing unitary inversion, which implements $f(U)=U^{-1}$ with multiple queries to $U$~\cite{chiribella2016optimal, quintino2019probabilistic,quintino2019reversing, sardharwalla2016universal, quintino2022deterministic, navascues2018resetting, trillo2020translating, trillo2023universal, yoshida2023reversing, chen2024quantum, mo2025parameterized}.
Beyond quantum control, unitary inversion also has an application in the measurement of the out-of-time-order correlators~\cite{larkin1969quasiclassical, maldacena2016bound, garttner2017measuring, li2017measuring} and quantum singular value transformations~\cite{low2017optimal, low2019hamiltonian, gilyen2019quantum, martyn2021grand}.
In quantum functional programming, one would like to execute a unitary operation $f(U)$ for a quantum program representing $U$ without looking into the detail of the quantum program.
In these applications, the unknown nature of the input unitary is essential for ensuring both the \emph{blindness} of the input and the \emph{universality} of the algorithm.

Similarly to unknown quantum states, no-go theorems are known for several transformations of a unitary operation with a single query of the black-box unitary operation  \cite{chiribella2008optimal, chiribella2016optimal}.  
One way to circumvent this problem is to consider the algorithms using multiple queries of the black-box unitary operations.
However, deterministic and exact transformations of unknown unitary operations were considered to be impossible with finite queries since implementing such transformations was believed to require exact knowledge about at least a part of the unknown unitary operations, namely, the exact value of at least one of the parameters of the unitary operation. To obtain such an exact value via process tomography \cite{chuang1997prescription, baldwin2014quantum, haah2023query, yang2020optimal}, 
an infinite number of queries is necessary.
Thus, many previous works focus on the investigation of go and no-go results of probabilistic or approximate transformations \cite{dong2019controlled, quintino2019probabilistic,quintino2019reversing, chiribella2016optimal, chiribella2008optimal, bisio2010optimal, sedlak2019optimal, yang2019optimal, sedlak2020probabilistic, bisio2014optimal, dur2015deterministic, chiribella2015universal, soleimanifar2016nogo, ebler2023optimal, araujo2014quantum, bisio2016quantum, dong2021success, sardharwalla2016universal, quintino2022deterministic, navascues2018resetting, trillo2020translating, trillo2023universal}.

Contrary to intuition, recent works \cite{yoshida2023reversing, chen2024quantum, miyazaki2019complex} have proven that deterministic and exact transformations of an unknown unitary operation to its complex conjugation, inversion, and transposition can be achieved with a \textit{finite} number of queries of the black-box unitary operation.  In addition, the existence of \textit{catalytic transformations} was found for unitary inversion \cite{yoshida2023reversing}.  These discoveries suggest that these transformations of an unknown unitary operation can be achieved fully within a quantum regime with a finite number of multiple queries without extracting classical knowledge about the black-box unitary operation.  That is, the queries served as a resource solely for transformation, not for extracting classical knowledge. Further, such a resource can be catalytic for some transformations. The lower bound of the number of queries characterizes the resource required for each transformation. 

However, no-go theorems for deterministic and exact transformations for a $d$-dimensional unknown unitary operation are still missing in general, and thus, the analytic lower bounds were not established except for unitary complex conjugation for which the tight lower bound $d-1$ is proven in \cite{quintino2019reversing}, and nonlinear transformations such as unitary controlization, which requires an infinite number of queries.
No-go theorems for parallel implementations are obtained for unitary transposition and inversion in Ref.~\cite{quintino2019probabilistic}, but showing no-go theorems for sequential implementations remains a hard task.
Numerical lower bounds for unitary inversion and transposition are obtained to be 4 for $d=2$ \cite{yoshida2023reversing,grinko2024linear}, but it is difficult to extend to general $d$ due to the complexity of the problem.  Regarding catalytic transformations, no condition for catalytic transformations for unknown unitary operations was known.

In this Letter, we provide a general framework for deriving no-go theorems for deterministic and exact transformations of a $d$-dimensional unknown unitary operation $U$ given as a differentiable function $f(U)$ mapping to another $d$-dimensional unitary operation.  From the no-go theorems, we obtain a lower bound of the required number of queries of the unitary operation to implement a function $f$ deterministically and exactly in terms of semidefinite programming (SDP).
Our lower bound is tight for unitary complex conjugation [$f(U)=U^*$] and asymptotically tight for unitary inversion [$f(U)=U^{-1}$].
Our lower bound also holds when the input unitary operation is an unknown logarithmic-depth unitary.
We have also shown the relationship between the tightness of the SDP and the non-existence of catalytic transformations.
Finally, we present extensions of our framework to relaxed requirements where the input unitary operation is partially known or the transformation is implemented exactly but probabilistically.

\textit{General lower bound for the query complexity of functions of unitary operations.}---
Within a quantum circuit model of quantum computation, a transformation $f(U)$ of a $d$-dimensional unknown unitary operation $U \in \SU(d)$ to another unitary $f(U) \in \SU(d)$ is deterministically and exactly possible with $N$ queries of the black-box unitary operation $U$ if such a transformation can be implemented by a fixed-order quantum circuit (also known as a quantum comb \cite{chiribella2008quantum}) including $N$ queries to $U$ in the middle of the quantum circuit as shown in Fig.~\ref{fig::lower_bound_main}.
We call the minimum number of the queries as the \textit{query complexity} of $f$.  If the deterministic and exact implementation of $f$ is impossible with finite queries, the query complexity of $f$ is defined as $\infty$.
When the query complexity of a function $f$ is shown to be larger than or equal to a number $N$, then a no-go theorem forbidding deterministic and exact implementation of $f$ with a query less than $N$ is derived.

\begin{figure}[tb]
    \centering
    \includegraphics[width=\linewidth]{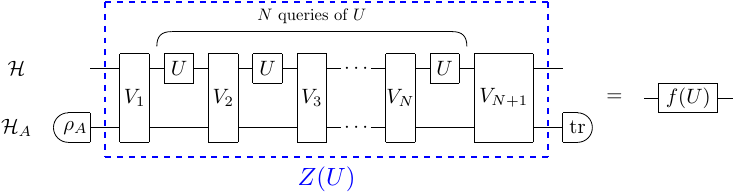}
    \caption{The quantum circuit implementing deterministic and exact transformation $f(U)$ for a black-box unitary operation $U$ with $N$ queries to $U$, where $\rho_A$ is a fixed state of the auxiliary system, and $V_1,\ldots ,V_{N+1}$ are unitary operations. $Z(U)$ is the unitary operation corresponding to the circuit without $\rho_A$ and tracing out.
    }
    \label{fig::lower_bound_main}
\end{figure}

Note that in the context of the property testing, it is often considered the situation where the queries of $U^{-1}$ and/or the controlled unitary operation $\mathrm{ctrl}-U$ can also be applicable in addition to the queries of $U$ \cite{montanaro2013survey, harrow2017sequential}. However, the exact transformation of $U$ to $U^{-1}$ requires at least $d^2$ queries (as we will show in this letter), and that of $U$ to $\mathrm{ctrl}-U$ is impossible with finite queries (query complexity $\infty$) \cite{gavorova2024topological} for an unknown $U$.  We choose a setting that only the black-box unitary operation $U$ can be used in the protocol to evaluate the query complexity of $f(U)$.

\begin{table*}[tb]
    \centering
    \begin{tabular}{c|c|c|c|c}
        \hline\hline
         &\multicolumn{2}{c|}{Lower bound}&\multicolumn{2}{c}{Minimum known}\\
         \cline{2-5}
         &previous methods &our method &$d=2$ &$d\geq3$  \\
         \hline
         $f(U)=U^{-1}$ &{$4^*$} ($d=2$ \cite{yoshida2023reversing}), {$6^*$} ($3\leq d\leq 7$ \cite{yoshida2023reversing}), 
         $d-1$ ($d\geq 8$ \cite{quintino2019probabilistic}) & $\bm{d^2}$  & $4^*$ \cite{yoshida2023reversing}& $\sim (\pi /2)d^2$ \cite{chen2024quantum}  \\
         \hline
         $f(U)=U^{T}$&
         {$4^*$} ($d=2$), $5^*$ ($d\geq 3$) \cite{grinko2024linear}
         &$\bm{4}\ (d=2),\bm{d+3}\ (d\geq 3)$&$4$* \cite{grinko2024linear}& $\sim (\pi/2)d^2$ \cite{chen2024quantum}  \\
         \hline
         $f(U)=U^{*}$&$d-1$ \cite{quintino2019probabilistic}&$d-1$&\multicolumn{2}{c}{$d-1$ \cite{miyazaki2019complex}}    \\
    \hline \hline
    \end{tabular}
    \caption{Comparison of the lower bound of the query complexity of the deterministic and exact implementation of $f(U)$ for a $d$-dimensional unitary $U$ obtained by Cor.~\ref{cor::inversion-transposistion-conjugation} and the minimum number of queries achievable by the algorithms given in \cite{yoshida2023reversing,chen2024quantum, miyazaki2019complex, grinko2024linear}.
    The lower bounds obtained numerically are shown with an asterisk ${}^*$.}
    \label{tab::thm_vs_known}
\end{table*}

We first derive a lower bound of a query complexity of a \emph{differentiable} function $f$.
We call a function $f:\SU(d)\to\SU(d)$ is differentiable around $U_0\in\SU(d)$ if there exists a linear map $g_{U_0}:\su(d)\to \su(d)$ such that
\begin{align}
\label{eq:differentiation}
    f(e^{i\epsilon H}U_0) = f(U_0) e^{i\epsilon g_{U_0}(H) + O(\epsilon^2)} \quad \forall H\in\su(d),
\end{align}
where $\su(d)$ is given by $\su(d) = \{H\in\mcL(\CC^d)| \tr(H) = 0, H=H^\dagger\}$.
We define the Choi operator \cite{choi1975completely, jamiolkowski1972linear} of $g_{U_0}$ by
\begin{align}
\label{eq::J_g}
    J_{g_{U_0}} \coloneqq \sum_{j=1}^{d^2-1} G_j^* \otimes g_{U_0}(G_j),
\end{align}
where $\{G_j\}_j$ is an orthonormal basis (in terms of the Hilbert-Schmidt inner product) of $\su(d)$, and $*$ is the complex conjugation in the computational basis.
Then, we show the following theorem on the query complexity of $f$.
\begin{Theorem}\label{th::main_sdp}
    Given any differentiable function $f: \SU(d) \to \SU(d)$, the query complexity of $f$ is at least the solution of the following SDP:
\begin{align}
\begin{split}
    &\min \tr \beta_{U_0}\\
    \text{s.t. }&\beta_{U_0} \in \mcL(\CC^d),\\
    &J_{g_{U_0}} + \beta_{U_0} \otimes I \geq 0,
\end{split}
\label{eq:sdp}
\end{align}
where $J_{g_{U_0}}$ is defined in Eq.~\eqref{eq::J_g}.
\end{Theorem}
\begin{proof}[Proof sketch]
Any fixed-order circuit transforming an arbitrary unitary $U$ with $N$ queries can be represented by the quantum circuit using unitary operations $V_1, \ldots, V_{N+1}$ and a fixed state $\rho_A$ of the auxiliary system, as shown in Fig.~\ref{fig::lower_bound_main} \cite{chiribella2008quantum}.
We define the unitary operator $Z(U)$ as shown in Fig.~\ref{fig::lower_bound_main}.
We choose $U=e^{i\epsilon H} U_0$ where $H\in \su (d)$ is an Hermitian operator, $\epsilon$ is a real parameter, and $U_0$ is an arbitrarily unitary operator.
By considering the differentiation of the unitary operator $Z(e^{i\epsilon H} U_0)$ in terms of $\epsilon$ around $\epsilon=0$, we obtain $\mcE_{U_0} (H)$ represented by a linear map $\mcE_{U_0}$, which is a completely positive map defined by $V_1,\ldots ,V_{N+1}$ and satisfies $\mcE_{U_0}(I)=NI$.
On the other hand, the differentiation of $f(e^{i\epsilon H} U_0)$ is determined by $f$ as shown in Eq.~(\ref{eq:differentiation}).
The equality for all $H\in \su (d)$ identifies $\mcE_{U_0}$ up to some degree of freedom. When $N$ is too small, the resulting $\mcE_{U_0}$ cannot be taken completely positive. Thus, the corresponding circuit with $N$ queries to $U$ does not exist, providing a no-go theorem. This validity condition of $N$ is represented as the SDP constraints in Eq.~(\ref{eq:sdp}). See Appendix~\ref{app::theo1} of the Supplemental Material (SM)~\cite{supple} for the details of the proof.
\end{proof}

Since the proof of Thm.~\ref{th::main_sdp} relies on the differentiation, the lower bound of the query complexity shown in Thm.~\ref{th::main_sdp} holds even if we restrict the input unitary operation to be within the set $\{e^{iP_j \theta}\}_{\theta\in (-\epsilon, \epsilon), j}$ for a basis $\{P_j\}_j$ of $\su(d)$ and a sufficiently small $\epsilon>0$, but $j$ and $\theta$ are unknown (see Appendix~A of the SM~\cite{supple} for the details).
For $n$-qubit systems ($d = 2^n$), the set of Pauli operators $\{I, X, Y, Z\}^{\otimes n} \setminus \{I^{\otimes n}\}$ is a basis of $\su(d)$.
The Pauli rotation $e^{i P_j \theta}$ can be implemented by a single-qubit $X$-rotation, two multi-target-CNOT gates (or quantum fanout gates), and two layers of single-qubit Clifford gates.
Since the multi-target-CNOT gate can be implemented in depth $O(\log n)$ using the devide-and-conquer method~\cite{fang2003quantum}, the query complexity lower bound in Thm.~\ref{th::main_sdp} also holds for the case where the input unitary operation is restricted to be within the logarithmic-depth unitaries.
By applying Thm.~\ref{th::main_sdp} to unitary inversion, transposition, and complex conjugation, we obtain the exponential hardness of these transformations for the logarithmic-depth $n$-qubit unitary operations.

\begin{Corollary}
    \label{cor:logdepth}
    The query complexities of unitary inversion, transposition, and complex conjugation for the unknown logarithmic-depth $n$-qubit unitary are given by $\exp[\Theta(n)]$.
\end{Corollary}
\begin{proof}[Proof sketch]
    The SDP in Thm.~\ref{th::main_sdp} yields $d^2-1$ for unitary inversion, $d+1$ for unitary transposition and $d-1$ for unitary complex conjugation, i.e., the query complexities are given by $\exp[\Omega(n)]$ for $n$-qubit systems.
    These lower bounds are tight since the query complexities of these transformations for general $n$-qubit unitary operations are given by $\exp[O(n)]$ as shown in Refs.~\cite{yoshida2023reversing, chen2024quantum, miyazaki2019complex}.
    See Appendix~\ref{appendix:logdepth} of the SM~\cite{supple} for the details of the proof.
\end{proof}

For the unitary inversion and transposition, we can further improve the lower bounds by considering an extra argument based on the fact that the conditions used for the derivation of the SDP should hold for \textit{all} $U_0\in\SU(d)$, leading to the following corollary.
See Appendix~\ref{app:prooftable} of the SM~\cite{supple} for the details of the proof.
Note that this corollary for the unitary inversion and transposition is not shown for the case when we restrict the input unitary operation to be within the logarithmic-depth unitaries.

\begin{Corollary}
\label{cor::inversion-transposistion-conjugation}
    The query complexity of unitary inversion $[f(U)=U^{-1}]$, unitary transposition $[f(U)=U^T]$, and unitary complex conjugation $[f(U)=U^*]$ for $U\in \SU(d)$ are lower bounded by
    \begin{align}
        \begin{cases}
            d^2 & [f(U)=U^{-1}],\\
            4 & [f(U)=U^T, d=2],\\
            d+3 & [f(U)=U^T, d\geq 3],\\
            d-1 & [f(U)=U^*].
        \end{cases}
    \end{align}
\end{Corollary}

\textit{Comparison with previous works}.---
Previously known analytical lower bounds were $d-1$ for unitary inversion and $2$ for unitary transposition \cite{quintino2019probabilistic}.
They were obtained by the polynomial degree analysis or the Fourier series analysis.
These lower bounds were strictly smaller than the minimum number of queries required to implement unitary inversion or transposition in the $d=2$ case, which is numerically shown to be $4$ for unitary inversion \cite{yoshida2023reversing}, and for unitary transposition \cite{grinko2024linear}.

Our lower bounds are tight at $d=2$ and scale at the same rate $O(d^2)$ as the number of queries obtained by the recently discovered algorithm \cite{chen2024quantum} for unitary inversion.  Therefore, the analytical optimality of both the qubit-unitary inversion algorithm in \cite{yoshida2023reversing}, and the asymptotic optimality of the algorithm in \cite{chen2024quantum} are obtained from our result.   The unitary inversion algorithm in \cite{chen2024quantum} can be modified to a unitary transposition algorithm by swapping $U$ and $U^*$ in the algorithm. By implementing $U^*$ by using $d-1$ calls of $U$ \cite{miyazaki2019complex}, the modified algorithm implements unitary transposition by using $O(d^2)$ queries of $U$.
In contrast, our bound for unitary transposition scales only linearly on $d$, which indicates the possibility of an asymptotically more efficient algorithm.

We summarize the previously known lower bounds and the ones shown in Cor.~\ref{cor::inversion-transposistion-conjugation}, and minimum queries achieved by proposed algorithms in Tab.~\ref{tab::thm_vs_known} for the deterministic and exact transformations, unitary inversion $f(U)=U^{-1}$, transposition $f(U)=U^T$, and complex conjugation $f(U)=U^*$. 
The lower bounds shown in Tab. \ref{tab::thm_vs_known} are not guaranteed to be tight in general. 
Nevertheless, they are tight in all three transformations for $d=2$ and in unitary complex conjugation for general $d$.
Our lower bound also gives a matching lower bound $\Omega(d^2)$ for unitary inversion achievable with $O(d^2)$ queries.
This observation implies that our method potentially provides sufficiently tight bounds for certain types of $f$.

\textit{Necessary condition for the existence of catalytic transformation.}---
In deterministic and exact unitary inversion for $d=2$ \cite{yoshida2023reversing}, a novel property of ``catalytic'' transformation is observed. In short, the algorithm proposed in \cite{yoshida2023reversing} implements unitary inversion $U^{-1}$ using three queries to the black-box unitary $U$ and one query to the ``catalytic state'' which is generated using one query to $U$. While implementation of a single output of $U^{-1}$ requires four queries in total, the same catalytic state is output in the auxiliary system  (which is why the adjective ``catalytic'' is used), thus additional production of $U^{-1}$ requires only three more queries to $U$. More generally, implementation of $n$ copies of $U^{-1}$ requires $3n+1$ queries.
Even though the catalytic property of unitary transformation algorithms enhances their applicability by reducing the asymptotic cost, we can show the non-existence of optimal and catalytic algorithms for some transformations by following the theorem.
\begin{Theorem}\label{th::catal_cond}
    When the SDP solution $N$ of Eq.~(\ref{eq:sdp}) is tight for a function $f$ satisfying $f(U_0)=I$, then the optimal and catalytic algorithm for implementing $f$ does not exist. 
\end{Theorem}
\begin{proof}[{Proof}]
We show this theorem from Thm.~\ref{th::main_sdp} by comparing the derivative $g_{U_0}(H)$ in Eq.~(\ref{eq:differentiation}) for $U\mapsto f(U)$ and $U\mapsto f(U)^n$.
Let us define $g'_1(H)$ and $g'_n(H)$ by $g_{U_0}(H)$ for $f(U)$ and $f(U)^{n}$, respectively. Using the product rule (or the Leibniz rule) of differentiation, $g'_n(H)$ is expressed in terms of $g'_1(H)$ as
\begin{align}
    g'_n(H)&=\left.
        -i\frac{\rm d}{{\rm d}\epsilon} 
    \right|_{\epsilon = 0}
        [f(U_0)^{-n}f(e^{i\epsilon H}U_0)^n]
    \nonumber\\
    &= \sum_{k=0}^{n-1}f(U_0)^{-k} g'_{1}(H) f(U_0)^k
    \nonumber\\
    &=n g'_1(H),
\end{align}
thus the solution of the SDP given in Eq.~(\ref{eq:sdp}) is $nN$ for $f(U)^n$. If a catalytic transformation is possible, there exists a catalytic state that is prepared by the first run of the protocol and makes the further run of the protocol with $M$ queries of the input unitary for $M<N$.
Then, we can implement $f(U)^n$ with $N+(n-1)M$ queries, which is smaller than $nN$ and contradicts the SDP solution.
\end{proof}

Theorem~\ref{th::catal_cond} shows that the optimal algorithm for unitary complex conjugation is not catalytic, since the lower bound $d-1$ is tight for unitary complex conjugation.
We also prove that the optimal algorithm for the unitary iteration $U\mapsto U^n$ for a positive integer $n$ is not catalytic since the SDP solution for unitary iteration is $n$ as shown in Appendix~\ref{app:prooftable} of the SM~\cite{supple}, which is tight since the consecutive application of $U$ for $n$ times implements $U^n$.
In contrast, the SDP solutions for unitary inversion and transposition ($d^2-1$ and $d+1$, respectively) are strictly smaller than the numbers given by Cor.~\ref{cor::inversion-transposistion-conjugation}. Thus, they do not satisfy the assumption of Thm.~\ref{th::catal_cond} indicating the possibility of catalytic algorithms.

\begin{figure}
    \centering
    \includegraphics[width=0.8\linewidth]{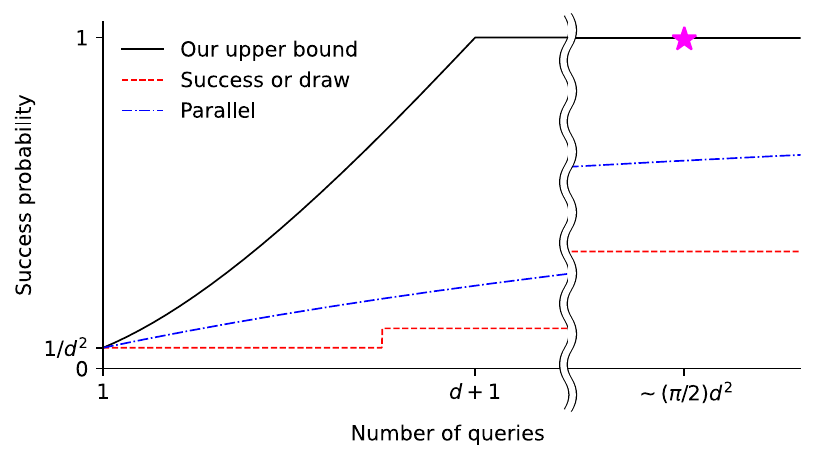}
    \caption{Summary of upper (``Our upper bound'') and lower (other lines) bounds of the success probability of unitary transposition. ``Our upper bound'' shows the analytical solution of the SDP for the probabilistic transformation. ``Success or draw'' and ``Parallel'' refer to the lower bounds corresponding to the success probability of protocols given in Sec.~E and Thm.~2 of \cite{quintino2019probabilistic}, respectively. Magenta star shows the number of queries $\sim (\pi /2)d^2$ required in the deterministic exact transposition algorithm given by modifying the algorithm in \cite{chen2024quantum}.}
    \label{fig::prob_trans_graph}
\end{figure}

\textit{Extension of Thm.~\ref{th::main_sdp} to relaxed situations.---}
Theorem~\ref{th::main_sdp} shows a lower bound of query complexity for a deterministic and exact implementation of a completely unknown unitary operation $U$.
However, this situation may be too restricted in a realistic setting.
For instance, partial information of $U$ may be available, and probabilistic implementation may be enough for a practical application.
We extend Thm.~\ref{th::main_sdp} to these relaxed situations.

We first consider the partially-known situation where a partial information of $U$ is given in terms of a Lie subgroup $S$ of $\SU(d)$, i.e., $U$ is promised to satisfy $U\in S$.
In this case, the linear map $\mcE_{U_0}$ is only determined on the Hamiltonian $H$ within its Lie algebra $\mathfrak{s}$ and consequently, the resulting SDP will have a solution which is smaller than or equal to the solution for the $\SU (d)$ case.
We give the SDP for this partially-known situation and show a lower bound for implementing the unitary inversion in ${\rm SO}(d)$, which is given by $d-1$ and tight for $d=2$.
This suggests a possibility to reduce the query complexity by using a partial knowledge about the input unitary operation.
See Appendix~\ref{app::subgroup} of the SM~\cite{supple} for the detail.

For the probabilistic situation, we provide a lower bound on query complexity to implement $f(U)$ with a success probability above $p(U)$ with a differentiable function $p: \SU(d)\to \RR_{\geq 0}$ in terms of SDP.
We show the upper bound of the success probability of probabilistic exact unitary transposition given by
\begin{align}\label{eq::trans_main_sucupper}
    p_{\rm trans}(U_0)\leq \left(\frac{d}{((d^2-1)/N)+1}\right)^2
\end{align}
for an arbitrary $U_0\in \SU(d)$, as shown in Fig.~\ref{fig::prob_trans_graph}.
This upper bound reproduces the tight upper bound $1/d^2$ of the success probability $p(U)$ of unitary transposition for $N=1$, which is achieved by the gate teleportation-based algorithm \cite{quintino2019probabilistic}.
In addition, this upper bound shows that the success probability $p(U_0)$ for a fixed number of queries tends to $0$ in the limit of $d\to \infty$ for any $N$.
Note that this upper bound is obtained only using the property of function $p(\cdot)$ in a neighbor of an arbitrary unitary $U_0$, thus our lower bound holds even if we allow low success probability outside the neighbor of $U_0$.
See Appendix~\ref{app::prob} of the SM~\cite{supple} for the details. 

\textit{Conclusion.}---
This work derives a lower bound for the query complexity of any differentiable function $f:\SU(d)\to\SU(d)$ in terms of SDP.
This framework shows the exponential hardness of unitary inversion, transposition, and complex conjugation even if the input $n$-qubit unitary is an unknown logarithmic-depth unitary.
The lower bounds for unitary inversion and unitary transposition are tight for $d=2$ and asymptotically tight for unitary inversion.
This framework also gives a necessary condition for a function $f$ to have optimal and catalytic transformations of a unitary operation, which excludes the possibility of catalytic transformation for unitary complex conjugation.
We also provide a generalization of our framework to partially-known and probabilistic situations.
Our work sets a line in between ``what is possible'' and ``what is impossible'' in transformation of unitary operations, which will serve as a guide line for improvement of existing protocols (e.g., unitary inversion~\cite{yoshida2023reversing, chen2024quantum}) by using partial knowledge about the input unitary operation or relaxing the protocols to be probabilistic ones.

This work introduces a new technique based on differentiation of unitary operations to analyze the query complexity of transformation of unitary operations (see Proof sketch of Thm.~\ref{th::main_sdp}), which is also used in the concurrent work of the authors~\cite{bavaresco2025simulating}.
We believe that this technique has more applications in the analysis of query complexity.
This work also opens up a new direction of research to extend this work as shown below.
First, we can consider extending the SDP to cover higher-order differentiation of $Z(U)$ to obtain tighter bounds, whereas we only consider first-order differentiation of $Z(U)$ in this work.
We can also consider a combination of our differentiation-based method with the polynomial degree-based method~\cite{quintino2019probabilistic}, which is used to prove the no-go results for probabilistic implementation of complex conjugation in less than $d-1$ queries, to obtain stronger no-go theorems.
Our lower bound for unitary transposition does not match with the best-known query complexity given in Ref.~\cite{chen2024quantum}.
We leave it a future work to investigate protocols achieving the query complexity given by lower bounds shown in this work.

{\it Note added.}---
Recently, Ref.~\cite{chen2025tight} shows a lower bound on the query complexity of approximate unitary inversion using a different technique than ours.

We would like to thank David Trillo Fern\'andez, Jisho Miyazaki, Marco T\'ulio Quintino, Jessica Bavaresco, Philip Taranto, Seiseki Akibue, and Hl\'er Kristj\'ansson for fruitful discussions. This work was supported by MEXT Quantum Leap Flagship Program (MEXT QLEAP)  JPMXS0118069605, JPMXS0120351339, Japan Society for the Promotion of Science (JSPS) KAKENHI Grant Number 23K21643 and 23KJ0734, FoPM, WINGS Program, the University of Tokyo, DAIKIN Fellowship Program, the University of Tokyo, and IBM Quantum.

\bibliography{main}

\begin{thebibliography}{55}%
\makeatletter
\providecommand \@ifxundefined [1]{%
 \@ifx{#1\undefined}
}%
\providecommand \@ifnum [1]{%
 \ifnum #1\expandafter \@firstoftwo
 \else \expandafter \@secondoftwo
 \fi
}%
\providecommand \@ifx [1]{%
 \ifx #1\expandafter \@firstoftwo
 \else \expandafter \@secondoftwo
 \fi
}%
\providecommand \natexlab [1]{#1}%
\providecommand \enquote  [1]{``#1''}%
\providecommand \bibnamefont  [1]{#1}%
\providecommand \bibfnamefont [1]{#1}%
\providecommand \citenamefont [1]{#1}%
\providecommand \href@noop [0]{\@secondoftwo}%
\providecommand \href [0]{\begingroup \@sanitize@url \@href}%
\providecommand \@href[1]{\@@startlink{#1}\@@href}%
\providecommand \@@href[1]{\endgroup#1\@@endlink}%
\providecommand \@sanitize@url [0]{\catcode `\\12\catcode `\$12\catcode `\&12\catcode `\#12\catcode `\^12\catcode `\_12\catcode `\%12\relax}%
\providecommand \@@startlink[1]{}%
\providecommand \@@endlink[0]{}%
\providecommand \url  [0]{\begingroup\@sanitize@url \@url }%
\providecommand \@url [1]{\endgroup\@href {#1}{\urlprefix }}%
\providecommand \urlprefix  [0]{URL }%
\providecommand \Eprint [0]{\href }%
\providecommand \doibase [0]{https://doi.org/}%
\providecommand \selectlanguage [0]{\@gobble}%
\providecommand \bibinfo  [0]{\@secondoftwo}%
\providecommand \bibfield  [0]{\@secondoftwo}%
\providecommand \translation [1]{[#1]}%
\providecommand \BibitemOpen [0]{}%
\providecommand \bibitemStop [0]{}%
\providecommand \bibitemNoStop [0]{.\EOS\space}%
\providecommand \EOS [0]{\spacefactor3000\relax}%
\providecommand \BibitemShut  [1]{\csname bibitem#1\endcsname}%
\let\auto@bib@innerbib\@empty
\bibitem [{\citenamefont {Wootters}\ and\ \citenamefont {Zurek}(1982)}]{wootters1982single}%
  \BibitemOpen
  \bibfield  {author} {\bibinfo {author} {\bibfnamefont {W.~K.}\ \bibnamefont {Wootters}}\ and\ \bibinfo {author} {\bibfnamefont {W.~H.}\ \bibnamefont {Zurek}},\ }\bibfield  {title} {\bibinfo {title} {A single quantum cannot be cloned},\ }\href {https://doi.org/10.1038/299802a0} {\bibfield  {journal} {\bibinfo  {journal} {Nature}\ }\textbf {\bibinfo {volume} {299}},\ \bibinfo {pages} {802} (\bibinfo {year} {1982})}\BibitemShut {NoStop}%
\bibitem [{\citenamefont {Bennett}\ and\ \citenamefont {Brassard}(2014)}]{bennett2014quantum}%
  \BibitemOpen
  \bibfield  {author} {\bibinfo {author} {\bibfnamefont {C.~H.}\ \bibnamefont {Bennett}}\ and\ \bibinfo {author} {\bibfnamefont {G.}~\bibnamefont {Brassard}},\ }\bibfield  {title} {\bibinfo {title} {Quantum cryptography: Public key distribution and coin tossing},\ }\href {https://doi.org/10.1016/j.tcs.2014.05.025} {\bibfield  {journal} {\bibinfo  {journal} {Theor. Comput. Sci.}\ }\textbf {\bibinfo {volume} {560}},\ \bibinfo {pages} {7} (\bibinfo {year} {2014})},\ \Eprint {https://arxiv.org/abs/2003.06557} {arXiv:2003.06557} \BibitemShut {NoStop}%
\bibitem [{\citenamefont {Werner}(1998)}]{werner1998optimal}%
  \BibitemOpen
  \bibfield  {author} {\bibinfo {author} {\bibfnamefont {R.~F.}\ \bibnamefont {Werner}},\ }\bibfield  {title} {\bibinfo {title} {Optimal cloning of pure states},\ }\href {https://doi.org/10.1103/PhysRevA.58.1827} {\bibfield  {journal} {\bibinfo  {journal} {Phys. Rev. A}\ }\textbf {\bibinfo {volume} {58}},\ \bibinfo {pages} {1827} (\bibinfo {year} {1998})}\BibitemShut {NoStop}%
\bibitem [{\citenamefont {Agrawal}\ and\ \citenamefont {Pati}(2002)}]{agrawal2002probabilistic}%
  \BibitemOpen
  \bibfield  {author} {\bibinfo {author} {\bibfnamefont {P.}~\bibnamefont {Agrawal}}\ and\ \bibinfo {author} {\bibfnamefont {A.~K.}\ \bibnamefont {Pati}},\ }\bibfield  {title} {\bibinfo {title} {Probabilistic quantum teleportation},\ }\href {https://doi.org/10.1016/S0375-9601%2802%2901383-X} {\bibfield  {journal} {\bibinfo  {journal} {Phys. Lett. A}\ }\textbf {\bibinfo {volume} {305}},\ \bibinfo {pages} {12} (\bibinfo {year} {2002})},\ \Eprint {https://arxiv.org/abs/quant-ph/0210004} {arXiv:quant-ph/0210004} \BibitemShut {NoStop}%
\bibitem [{\citenamefont {Bu\ifmmode~\check{z}\else \v{z}\fi{}ek}\ and\ \citenamefont {Hillery}(2000)}]{buzek2000optimal}%
  \BibitemOpen
  \bibfield  {author} {\bibinfo {author} {\bibfnamefont {V.}~\bibnamefont {Bu\ifmmode~\check{z}\else \v{z}\fi{}ek}}\ and\ \bibinfo {author} {\bibfnamefont {M.}~\bibnamefont {Hillery}},\ }\bibfield  {title} {\bibinfo {title} {Optimal manipulations with qubits: Universal quantum entanglers},\ }\href {https://doi.org/10.1103/PhysRevA.62.022303} {\bibfield  {journal} {\bibinfo  {journal} {Phys. Rev. A}\ }\textbf {\bibinfo {volume} {62}},\ \bibinfo {pages} {022303} (\bibinfo {year} {2000})}\BibitemShut {NoStop}%
\bibitem [{\citenamefont {Yang}\ \emph {et~al.}(2020{\natexlab{a}})\citenamefont {Yang}, \citenamefont {Renner},\ and\ \citenamefont {Chiribella}}]{yang2019optimal}%
  \BibitemOpen
  \bibfield  {author} {\bibinfo {author} {\bibfnamefont {Y.}~\bibnamefont {Yang}}, \bibinfo {author} {\bibfnamefont {R.}~\bibnamefont {Renner}},\ and\ \bibinfo {author} {\bibfnamefont {G.}~\bibnamefont {Chiribella}},\ }\bibfield  {title} {\bibinfo {title} {{Optimal Universal Programming of Unitary Gates}},\ }\href {https://doi.org/10.1103/PhysRevLett.125.210501} {\bibfield  {journal} {\bibinfo  {journal} {Phys. Rev. Lett.}\ }\textbf {\bibinfo {volume} {125}},\ \bibinfo {pages} {210501} (\bibinfo {year} {2020}{\natexlab{a}})},\ \Eprint {https://arxiv.org/abs/2007.10363} {arXiv:2007.10363} \BibitemShut {NoStop}%
\bibitem [{\citenamefont {Navascu\'es}(2018)}]{navascues2018resetting}%
  \BibitemOpen
  \bibfield  {author} {\bibinfo {author} {\bibfnamefont {M.}~\bibnamefont {Navascu\'es}},\ }\bibfield  {title} {\bibinfo {title} {{Resetting Uncontrolled Quantum Systems}},\ }\href {https://doi.org/10.1103/PhysRevX.8.031008} {\bibfield  {journal} {\bibinfo  {journal} {Phys. Rev. X}\ }\textbf {\bibinfo {volume} {8}},\ \bibinfo {pages} {031008} (\bibinfo {year} {2018})},\ \Eprint {https://arxiv.org/abs/1710.02470} {arXiv:1710.02470} \BibitemShut {NoStop}%
\bibitem [{\citenamefont {G{\"a}rttner}\ \emph {et~al.}(2017)\citenamefont {G{\"a}rttner}, \citenamefont {Bohnet}, \citenamefont {Safavi-Naini}, \citenamefont {Wall}, \citenamefont {Bollinger},\ and\ \citenamefont {Rey}}]{garttner2017measuring}%
  \BibitemOpen
  \bibfield  {author} {\bibinfo {author} {\bibfnamefont {M.}~\bibnamefont {G{\"a}rttner}}, \bibinfo {author} {\bibfnamefont {J.~G.}\ \bibnamefont {Bohnet}}, \bibinfo {author} {\bibfnamefont {A.}~\bibnamefont {Safavi-Naini}}, \bibinfo {author} {\bibfnamefont {M.~L.}\ \bibnamefont {Wall}}, \bibinfo {author} {\bibfnamefont {J.~J.}\ \bibnamefont {Bollinger}},\ and\ \bibinfo {author} {\bibfnamefont {A.~M.}\ \bibnamefont {Rey}},\ }\bibfield  {title} {\bibinfo {title} {Measuring out-of-time-order correlations and multiple quantum spectra in a trapped-ion quantum magnet},\ }\href {https://doi.org/10.1038/nphys4119} {\bibfield  {journal} {\bibinfo  {journal} {Nat. Phys.}\ }\textbf {\bibinfo {volume} {13}},\ \bibinfo {pages} {781} (\bibinfo {year} {2017})},\ \Eprint {https://arxiv.org/abs/1608.08938} {arXiv:1608.08938} \BibitemShut {NoStop}%
\bibitem [{\citenamefont {Li}\ \emph {et~al.}(2017)\citenamefont {Li}, \citenamefont {Fan}, \citenamefont {Wang}, \citenamefont {Ye}, \citenamefont {Zeng}, \citenamefont {Zhai}, \citenamefont {Peng},\ and\ \citenamefont {Du}}]{li2017measuring}%
  \BibitemOpen
  \bibfield  {author} {\bibinfo {author} {\bibfnamefont {J.}~\bibnamefont {Li}}, \bibinfo {author} {\bibfnamefont {R.}~\bibnamefont {Fan}}, \bibinfo {author} {\bibfnamefont {H.}~\bibnamefont {Wang}}, \bibinfo {author} {\bibfnamefont {B.}~\bibnamefont {Ye}}, \bibinfo {author} {\bibfnamefont {B.}~\bibnamefont {Zeng}}, \bibinfo {author} {\bibfnamefont {H.}~\bibnamefont {Zhai}}, \bibinfo {author} {\bibfnamefont {X.}~\bibnamefont {Peng}},\ and\ \bibinfo {author} {\bibfnamefont {J.}~\bibnamefont {Du}},\ }\bibfield  {title} {\bibinfo {title} {Measuring out-of-time-order correlators on a nuclear magnetic resonance quantum simulator},\ }\href {https://doi.org/10.1103/PhysRevX.7.031011} {\bibfield  {journal} {\bibinfo  {journal} {Phys. Rev. X}\ }\textbf {\bibinfo {volume} {7}},\ \bibinfo {pages} {031011} (\bibinfo {year} {2017})},\ \Eprint {https://arxiv.org/abs/1609.01246} {arXiv:1609.01246} \BibitemShut {NoStop}%
\bibitem [{\citenamefont {Bisio}\ and\ \citenamefont {Perinotti}(2019)}]{bisio2019theoretical}%
  \BibitemOpen
  \bibfield  {author} {\bibinfo {author} {\bibfnamefont {A.}~\bibnamefont {Bisio}}\ and\ \bibinfo {author} {\bibfnamefont {P.}~\bibnamefont {Perinotti}},\ }\bibfield  {title} {\bibinfo {title} {{Theoretical framework for higher-order quantum theory}},\ }\href {https://doi.org/10.1098/rspa.2018.0706} {\bibfield  {journal} {\bibinfo  {journal} {Proc. R. Soc. A}\ }\textbf {\bibinfo {volume} {475}},\ \bibinfo {pages} {20180706} (\bibinfo {year} {2019})},\ \Eprint {https://arxiv.org/abs/1806.09554} {arXiv:1806.09554} \BibitemShut {NoStop}%
\bibitem [{\citenamefont {Chiribella}\ and\ \citenamefont {Ebler}(2016)}]{chiribella2016optimal}%
  \BibitemOpen
  \bibfield  {author} {\bibinfo {author} {\bibfnamefont {G.}~\bibnamefont {Chiribella}}\ and\ \bibinfo {author} {\bibfnamefont {D.}~\bibnamefont {Ebler}},\ }\bibfield  {title} {\bibinfo {title} {Optimal quantum networks and one-shot entropies},\ }\href {https://doi.org/10.1088/1367-2630/18/9/093053} {\bibfield  {journal} {\bibinfo  {journal} {New J. Phys.}\ }\textbf {\bibinfo {volume} {18}},\ \bibinfo {pages} {093053} (\bibinfo {year} {2016})},\ \Eprint {https://arxiv.org/abs/1606.02394} {arXiv:1606.02394} \BibitemShut {NoStop}%
\bibitem [{\citenamefont {Quintino}\ \emph {et~al.}(2019{\natexlab{a}})\citenamefont {Quintino}, \citenamefont {Dong}, \citenamefont {Shimbo}, \citenamefont {Soeda},\ and\ \citenamefont {Murao}}]{quintino2019probabilistic}%
  \BibitemOpen
  \bibfield  {author} {\bibinfo {author} {\bibfnamefont {M.~T.}\ \bibnamefont {Quintino}}, \bibinfo {author} {\bibfnamefont {Q.}~\bibnamefont {Dong}}, \bibinfo {author} {\bibfnamefont {A.}~\bibnamefont {Shimbo}}, \bibinfo {author} {\bibfnamefont {A.}~\bibnamefont {Soeda}},\ and\ \bibinfo {author} {\bibfnamefont {M.}~\bibnamefont {Murao}},\ }\bibfield  {title} {\bibinfo {title} {{Probabilistic exact universal quantum circuits for transforming unitary operations}},\ }\href {https://doi.org/10.1103/PhysRevA.100.062339} {\bibfield  {journal} {\bibinfo  {journal} {Phys. Rev. A}\ }\textbf {\bibinfo {volume} {100}},\ \bibinfo {pages} {062339} (\bibinfo {year} {2019}{\natexlab{a}})},\ \Eprint {https://arxiv.org/abs/1909.01366} {arXiv:1909.01366} \BibitemShut {NoStop}%
\bibitem [{\citenamefont {Quintino}\ \emph {et~al.}(2019{\natexlab{b}})\citenamefont {Quintino}, \citenamefont {Dong}, \citenamefont {Shimbo}, \citenamefont {Soeda},\ and\ \citenamefont {Murao}}]{quintino2019reversing}%
  \BibitemOpen
  \bibfield  {author} {\bibinfo {author} {\bibfnamefont {M.~T.}\ \bibnamefont {Quintino}}, \bibinfo {author} {\bibfnamefont {Q.}~\bibnamefont {Dong}}, \bibinfo {author} {\bibfnamefont {A.}~\bibnamefont {Shimbo}}, \bibinfo {author} {\bibfnamefont {A.}~\bibnamefont {Soeda}},\ and\ \bibinfo {author} {\bibfnamefont {M.}~\bibnamefont {Murao}},\ }\bibfield  {title} {\bibinfo {title} {{Reversing unknown quantum transformations: Universal quantum circuit for inverting general unitary operations}},\ }\href {https://doi.org/10.1103/PhysRevLett.123.210502} {\bibfield  {journal} {\bibinfo  {journal} {Phys. Rev. Lett.}\ }\textbf {\bibinfo {volume} {123}},\ \bibinfo {pages} {210502} (\bibinfo {year} {2019}{\natexlab{b}})},\ \Eprint {https://arxiv.org/abs/1810.06944} {arXiv:1810.06944} \BibitemShut {NoStop}%
\bibitem [{\citenamefont {Sardharwalla}\ \emph {et~al.}(2016)\citenamefont {Sardharwalla}, \citenamefont {Cubitt}, \citenamefont {Harrow},\ and\ \citenamefont {Linden}}]{sardharwalla2016universal}%
  \BibitemOpen
  \bibfield  {author} {\bibinfo {author} {\bibfnamefont {I.~S.}\ \bibnamefont {Sardharwalla}}, \bibinfo {author} {\bibfnamefont {T.~S.}\ \bibnamefont {Cubitt}}, \bibinfo {author} {\bibfnamefont {A.~W.}\ \bibnamefont {Harrow}},\ and\ \bibinfo {author} {\bibfnamefont {N.}~\bibnamefont {Linden}},\ }\bibfield  {title} {\bibinfo {title} {{Universal refocusing of systematic quantum noise}},\ }\Eprint {https://arxiv.org/abs/1602.07963} {arXiv:1602.07963}  (\bibinfo {year} {2016})\BibitemShut {NoStop}%
\bibitem [{\citenamefont {Quintino}\ and\ \citenamefont {Ebler}(2022)}]{quintino2022deterministic}%
  \BibitemOpen
  \bibfield  {author} {\bibinfo {author} {\bibfnamefont {M.~T.}\ \bibnamefont {Quintino}}\ and\ \bibinfo {author} {\bibfnamefont {D.}~\bibnamefont {Ebler}},\ }\bibfield  {title} {\bibinfo {title} {{Deterministic transformations between unitary operations: Exponential advantage with adaptive quantum circuits and the power of indefinite causality}},\ }\href {https://doi.org/10.22331/q-2022-03-31-679} {\bibfield  {journal} {\bibinfo  {journal} {Quantum}\ }\textbf {\bibinfo {volume} {6}},\ \bibinfo {pages} {679} (\bibinfo {year} {2022})},\ \Eprint {https://arxiv.org/abs/2109.08202} {arXiv:2109.08202} \BibitemShut {NoStop}%
\bibitem [{\citenamefont {Trillo}\ \emph {et~al.}(2020)\citenamefont {Trillo}, \citenamefont {Dive},\ and\ \citenamefont {Navascu{\'e}s}}]{trillo2020translating}%
  \BibitemOpen
  \bibfield  {author} {\bibinfo {author} {\bibfnamefont {D.}~\bibnamefont {Trillo}}, \bibinfo {author} {\bibfnamefont {B.}~\bibnamefont {Dive}},\ and\ \bibinfo {author} {\bibfnamefont {M.}~\bibnamefont {Navascu{\'e}s}},\ }\bibfield  {title} {\bibinfo {title} {Translating uncontrolled systems in time},\ }\href {https://doi.org/10.22331/q-2020-12-15-374} {\bibfield  {journal} {\bibinfo  {journal} {Quantum}\ }\textbf {\bibinfo {volume} {4}},\ \bibinfo {pages} {374} (\bibinfo {year} {2020})},\ \Eprint {https://arxiv.org/abs/1903.10568} {arXiv:1903.10568} \BibitemShut {NoStop}%
\bibitem [{\citenamefont {Trillo}\ \emph {et~al.}(2023)\citenamefont {Trillo}, \citenamefont {Dive},\ and\ \citenamefont {Navascu\'es}}]{trillo2023universal}%
  \BibitemOpen
  \bibfield  {author} {\bibinfo {author} {\bibfnamefont {D.}~\bibnamefont {Trillo}}, \bibinfo {author} {\bibfnamefont {B.}~\bibnamefont {Dive}},\ and\ \bibinfo {author} {\bibfnamefont {M.}~\bibnamefont {Navascu\'es}},\ }\bibfield  {title} {\bibinfo {title} {{Universal Quantum Rewinding Protocol with an Arbitrarily High Probability of Success}},\ }\href {https://doi.org/10.1103/PhysRevLett.130.110201} {\bibfield  {journal} {\bibinfo  {journal} {Phys. Rev. Lett.}\ }\textbf {\bibinfo {volume} {130}},\ \bibinfo {pages} {110201} (\bibinfo {year} {2023})},\ \Eprint {https://arxiv.org/abs/2205.01131} {arXiv:2205.01131} \BibitemShut {NoStop}%
\bibitem [{\citenamefont {Yoshida}\ \emph {et~al.}(2023)\citenamefont {Yoshida}, \citenamefont {Soeda},\ and\ \citenamefont {Murao}}]{yoshida2023reversing}%
  \BibitemOpen
  \bibfield  {author} {\bibinfo {author} {\bibfnamefont {S.}~\bibnamefont {Yoshida}}, \bibinfo {author} {\bibfnamefont {A.}~\bibnamefont {Soeda}},\ and\ \bibinfo {author} {\bibfnamefont {M.}~\bibnamefont {Murao}},\ }\bibfield  {title} {\bibinfo {title} {{Reversing Unknown Qubit-Unitary Operation, Deterministically and Exactly}},\ }\href {https://doi.org/10.1103/PhysRevLett.131.120602} {\bibfield  {journal} {\bibinfo  {journal} {Phys. Rev. Lett.}\ }\textbf {\bibinfo {volume} {131}},\ \bibinfo {pages} {120602} (\bibinfo {year} {2023})},\ \Eprint {https://arxiv.org/abs/2209.02907} {arXiv:2209.02907} \BibitemShut {NoStop}%
\bibitem [{\citenamefont {Chen}\ \emph {et~al.}(2024)\citenamefont {Chen}, \citenamefont {Mo}, \citenamefont {Liu}, \citenamefont {Zhang},\ and\ \citenamefont {Wang}}]{chen2024quantum}%
  \BibitemOpen
  \bibfield  {author} {\bibinfo {author} {\bibfnamefont {Y.-A.}\ \bibnamefont {Chen}}, \bibinfo {author} {\bibfnamefont {Y.}~\bibnamefont {Mo}}, \bibinfo {author} {\bibfnamefont {Y.}~\bibnamefont {Liu}}, \bibinfo {author} {\bibfnamefont {L.}~\bibnamefont {Zhang}},\ and\ \bibinfo {author} {\bibfnamefont {X.}~\bibnamefont {Wang}},\ }\bibfield  {title} {\bibinfo {title} {{Quantum Advantage in Reversing Unknown Unitary Evolutions}},\ }\Eprint {https://arxiv.org/abs/2403.04704} {arXiv:2403.04704}  (\bibinfo {year} {2024})\BibitemShut {NoStop}%
\bibitem [{\citenamefont {Mo}\ \emph {et~al.}(2025)\citenamefont {Mo}, \citenamefont {Zhang}, \citenamefont {Chen}, \citenamefont {Liu}, \citenamefont {Lin},\ and\ \citenamefont {Wang}}]{mo2025parameterized}%
  \BibitemOpen
  \bibfield  {author} {\bibinfo {author} {\bibfnamefont {Y.}~\bibnamefont {Mo}}, \bibinfo {author} {\bibfnamefont {L.}~\bibnamefont {Zhang}}, \bibinfo {author} {\bibfnamefont {Y.-A.}\ \bibnamefont {Chen}}, \bibinfo {author} {\bibfnamefont {Y.}~\bibnamefont {Liu}}, \bibinfo {author} {\bibfnamefont {T.}~\bibnamefont {Lin}},\ and\ \bibinfo {author} {\bibfnamefont {X.}~\bibnamefont {Wang}},\ }\bibfield  {title} {\bibinfo {title} {Parameterized quantum comb and simpler circuits for reversing unknown qubit-unitary operations},\ }\href {https://doi.org/10.1038/s41534-025-00979-1} {\bibfield  {journal} {\bibinfo  {journal} {npj Quantum Inf.}\ }\textbf {\bibinfo {volume} {11}},\ \bibinfo {pages} {32} (\bibinfo {year} {2025})},\ \Eprint {https://arxiv.org/abs/2403.03761} {arXiv:2403.03761} \BibitemShut {NoStop}%
\bibitem [{\citenamefont {Larkin}\ and\ \citenamefont {Ovchinnikov}(1969)}]{larkin1969quasiclassical}%
  \BibitemOpen
  \bibfield  {author} {\bibinfo {author} {\bibfnamefont {A.~I.}\ \bibnamefont {Larkin}}\ and\ \bibinfo {author} {\bibfnamefont {Y.~N.}\ \bibnamefont {Ovchinnikov}},\ }\bibfield  {title} {\bibinfo {title} {Quasiclassical method in the theory of superconductivity},\ }\href@noop {} {\bibfield  {journal} {\bibinfo  {journal} {Sov. Phys. JETP}\ }\textbf {\bibinfo {volume} {28}},\ \bibinfo {pages} {1200} (\bibinfo {year} {1969})}\BibitemShut {NoStop}%
\bibitem [{\citenamefont {Maldacena}\ \emph {et~al.}(2016)\citenamefont {Maldacena}, \citenamefont {Shenker},\ and\ \citenamefont {Stanford}}]{maldacena2016bound}%
  \BibitemOpen
  \bibfield  {author} {\bibinfo {author} {\bibfnamefont {J.}~\bibnamefont {Maldacena}}, \bibinfo {author} {\bibfnamefont {S.~H.}\ \bibnamefont {Shenker}},\ and\ \bibinfo {author} {\bibfnamefont {D.}~\bibnamefont {Stanford}},\ }\bibfield  {title} {\bibinfo {title} {A bound on chaos},\ }\href {https://doi.org/10.1007/JHEP08%282016%29106} {\bibfield  {journal} {\bibinfo  {journal} {J. High Energy Phys.}\ }\textbf {\bibinfo {volume} {2016}}\bibfield  {number} {\bibinfo  {number} { (8)},\ \bibinfo {pages} {1}},\ }\Eprint {https://arxiv.org/abs/1503.01409} {arXiv:1503.01409} \BibitemShut {NoStop}%
\bibitem [{\citenamefont {Low}\ and\ \citenamefont {Chuang}(2017)}]{low2017optimal}%
  \BibitemOpen
  \bibfield  {author} {\bibinfo {author} {\bibfnamefont {G.~H.}\ \bibnamefont {Low}}\ and\ \bibinfo {author} {\bibfnamefont {I.~L.}\ \bibnamefont {Chuang}},\ }\bibfield  {title} {\bibinfo {title} {{Optimal Hamiltonian simulation by quantum signal processing}},\ }\href {https://doi.org/10.1103/PhysRevLett.118.010501} {\bibfield  {journal} {\bibinfo  {journal} {Phys. Rev. Lett.}\ }\textbf {\bibinfo {volume} {118}},\ \bibinfo {pages} {010501} (\bibinfo {year} {2017})},\ \Eprint {https://arxiv.org/abs/1606.02685} {arXiv:1606.02685} \BibitemShut {NoStop}%
\bibitem [{\citenamefont {Low}\ and\ \citenamefont {Chuang}(2019)}]{low2019hamiltonian}%
  \BibitemOpen
  \bibfield  {author} {\bibinfo {author} {\bibfnamefont {G.~H.}\ \bibnamefont {Low}}\ and\ \bibinfo {author} {\bibfnamefont {I.~L.}\ \bibnamefont {Chuang}},\ }\bibfield  {title} {\bibinfo {title} {{Hamiltonian simulation by qubitization}},\ }\href {https://doi.org/10.22331/q-2019-07-12-163} {\bibfield  {journal} {\bibinfo  {journal} {Quantum}\ }\textbf {\bibinfo {volume} {3}},\ \bibinfo {pages} {163} (\bibinfo {year} {2019})},\ \Eprint {https://arxiv.org/abs/1610.06546} {arXiv:1610.06546} \BibitemShut {NoStop}%
\bibitem [{\citenamefont {Gily{\'e}n}\ \emph {et~al.}(2019)\citenamefont {Gily{\'e}n}, \citenamefont {Su}, \citenamefont {Low},\ and\ \citenamefont {Wiebe}}]{gilyen2019quantum}%
  \BibitemOpen
  \bibfield  {author} {\bibinfo {author} {\bibfnamefont {A.}~\bibnamefont {Gily{\'e}n}}, \bibinfo {author} {\bibfnamefont {Y.}~\bibnamefont {Su}}, \bibinfo {author} {\bibfnamefont {G.~H.}\ \bibnamefont {Low}},\ and\ \bibinfo {author} {\bibfnamefont {N.}~\bibnamefont {Wiebe}},\ }\bibfield  {title} {\bibinfo {title} {{Quantum singular value transformation and beyond: exponential improvements for quantum matrix arithmetics}},\ }in\ \href {https://doi.org/10.1145/3313276.3316366} {\emph {\bibinfo {booktitle} {Proceedings of the 51st Annual ACM SIGACT Symposium on Theory of Computing}}}\ (\bibinfo {year} {2019})\ pp.\ \bibinfo {pages} {193--204},\ \Eprint {https://arxiv.org/abs/1806.01838} {arXiv:1806.01838} \BibitemShut {NoStop}%
\bibitem [{\citenamefont {Martyn}\ \emph {et~al.}(2021)\citenamefont {Martyn}, \citenamefont {Rossi}, \citenamefont {Tan},\ and\ \citenamefont {Chuang}}]{martyn2021grand}%
  \BibitemOpen
  \bibfield  {author} {\bibinfo {author} {\bibfnamefont {J.~M.}\ \bibnamefont {Martyn}}, \bibinfo {author} {\bibfnamefont {Z.~M.}\ \bibnamefont {Rossi}}, \bibinfo {author} {\bibfnamefont {A.~K.}\ \bibnamefont {Tan}},\ and\ \bibinfo {author} {\bibfnamefont {I.~L.}\ \bibnamefont {Chuang}},\ }\bibfield  {title} {\bibinfo {title} {{Grand Unification of Quantum Algorithms}},\ }\href {https://doi.org/10.1103/PRXQuantum.2.040203} {\bibfield  {journal} {\bibinfo  {journal} {PRX Quantum}\ }\textbf {\bibinfo {volume} {2}},\ \bibinfo {pages} {040203} (\bibinfo {year} {2021})},\ \Eprint {https://arxiv.org/abs/2105.02859} {arXiv:2105.02859} \BibitemShut {NoStop}%
\bibitem [{\citenamefont {Chiribella}\ \emph {et~al.}(2008{\natexlab{a}})\citenamefont {Chiribella}, \citenamefont {D’Ariano},\ and\ \citenamefont {Perinotti}}]{chiribella2008optimal}%
  \BibitemOpen
  \bibfield  {author} {\bibinfo {author} {\bibfnamefont {G.}~\bibnamefont {Chiribella}}, \bibinfo {author} {\bibfnamefont {G.~M.}\ \bibnamefont {D’Ariano}},\ and\ \bibinfo {author} {\bibfnamefont {P.}~\bibnamefont {Perinotti}},\ }\bibfield  {title} {\bibinfo {title} {{Optimal Cloning of Unitary Transformation}},\ }\href {https://doi.org/10.1103/PhysRevLett.101.180504} {\bibfield  {journal} {\bibinfo  {journal} {Phys. Rev. Lett.}\ }\textbf {\bibinfo {volume} {101}},\ \bibinfo {pages} {180504} (\bibinfo {year} {2008}{\natexlab{a}})},\ \Eprint {https://arxiv.org/abs/0804.0129} {arXiv:0804.0129} \BibitemShut {NoStop}%
\bibitem [{\citenamefont {Chuang}\ and\ \citenamefont {Nielsen}(1997)}]{chuang1997prescription}%
  \BibitemOpen
  \bibfield  {author} {\bibinfo {author} {\bibfnamefont {I.~L.}\ \bibnamefont {Chuang}}\ and\ \bibinfo {author} {\bibfnamefont {M.~A.}\ \bibnamefont {Nielsen}},\ }\bibfield  {title} {\bibinfo {title} {Prescription for experimental determination of the dynamics of a quantum black box},\ }\href {https://doi.org/10.1080/09500349708231894} {\bibfield  {journal} {\bibinfo  {journal} {J. Mod. Opt.}\ }\textbf {\bibinfo {volume} {44}},\ \bibinfo {pages} {2455} (\bibinfo {year} {1997})},\ \Eprint {https://arxiv.org/abs/quant-ph/9610001} {arXiv:quant-ph/9610001} \BibitemShut {NoStop}%
\bibitem [{\citenamefont {Baldwin}\ \emph {et~al.}(2014)\citenamefont {Baldwin}, \citenamefont {Kalev},\ and\ \citenamefont {Deutsch}}]{baldwin2014quantum}%
  \BibitemOpen
  \bibfield  {author} {\bibinfo {author} {\bibfnamefont {C.~H.}\ \bibnamefont {Baldwin}}, \bibinfo {author} {\bibfnamefont {A.}~\bibnamefont {Kalev}},\ and\ \bibinfo {author} {\bibfnamefont {I.~H.}\ \bibnamefont {Deutsch}},\ }\bibfield  {title} {\bibinfo {title} {Quantum process tomography of unitary and near-unitary maps},\ }\href {https://doi.org/10.1103/PhysRevA.90.012110} {\bibfield  {journal} {\bibinfo  {journal} {Phys. Rev. A}\ }\textbf {\bibinfo {volume} {90}},\ \bibinfo {pages} {012110} (\bibinfo {year} {2014})},\ \Eprint {https://arxiv.org/abs/1404.2877} {arXiv:1404.2877} \BibitemShut {NoStop}%
\bibitem [{\citenamefont {Haah}\ \emph {et~al.}(2023)\citenamefont {Haah}, \citenamefont {Kothari}, \citenamefont {O’Donnell},\ and\ \citenamefont {Tang}}]{haah2023query}%
  \BibitemOpen
  \bibfield  {author} {\bibinfo {author} {\bibfnamefont {J.}~\bibnamefont {Haah}}, \bibinfo {author} {\bibfnamefont {R.}~\bibnamefont {Kothari}}, \bibinfo {author} {\bibfnamefont {R.}~\bibnamefont {O’Donnell}},\ and\ \bibinfo {author} {\bibfnamefont {E.}~\bibnamefont {Tang}},\ }\bibfield  {title} {\bibinfo {title} {Query-optimal estimation of unitary channels in diamond distance},\ }in\ \href {https://doi.org/10.1109/FOCS57990.2023.00028} {\emph {\bibinfo {booktitle} {2023 IEEE 64th Annual Symposium on Foundations of Computer Science (FOCS)}}}\ (\bibinfo {organization} {IEEE},\ \bibinfo {year} {2023})\ pp.\ \bibinfo {pages} {363--390},\ \Eprint {https://arxiv.org/abs/2302.14066} {arXiv:2302.14066} \BibitemShut {NoStop}%
\bibitem [{\citenamefont {Yang}\ \emph {et~al.}(2020{\natexlab{b}})\citenamefont {Yang}, \citenamefont {Renner},\ and\ \citenamefont {Chiribella}}]{yang2020optimal}%
  \BibitemOpen
  \bibfield  {author} {\bibinfo {author} {\bibfnamefont {Y.}~\bibnamefont {Yang}}, \bibinfo {author} {\bibfnamefont {R.}~\bibnamefont {Renner}},\ and\ \bibinfo {author} {\bibfnamefont {G.}~\bibnamefont {Chiribella}},\ }\bibfield  {title} {\bibinfo {title} {{Optimal Universal Programming of Unitary Gates}},\ }\href {https://doi.org/10.1103/PhysRevLett.125.210501} {\bibfield  {journal} {\bibinfo  {journal} {Phys. Rev. Lett.}\ }\textbf {\bibinfo {volume} {125}},\ \bibinfo {pages} {210501} (\bibinfo {year} {2020}{\natexlab{b}})},\ \Eprint {https://arxiv.org/abs/2007.10363} {arXiv:2007.10363} \BibitemShut {NoStop}%
\bibitem [{\citenamefont {Dong}\ \emph {et~al.}(2019)\citenamefont {Dong}, \citenamefont {Nakayama}, \citenamefont {Soeda},\ and\ \citenamefont {Murao}}]{dong2019controlled}%
  \BibitemOpen
  \bibfield  {author} {\bibinfo {author} {\bibfnamefont {Q.}~\bibnamefont {Dong}}, \bibinfo {author} {\bibfnamefont {S.}~\bibnamefont {Nakayama}}, \bibinfo {author} {\bibfnamefont {A.}~\bibnamefont {Soeda}},\ and\ \bibinfo {author} {\bibfnamefont {M.}~\bibnamefont {Murao}},\ }\bibfield  {title} {\bibinfo {title} {{Controlled quantum operations and combs, and their applications to universal controllization of divisible unitary operations}},\ }\Eprint {https://arxiv.org/abs/1911.01645} {arXiv:1911.01645}  (\bibinfo {year} {2019})\BibitemShut {NoStop}%
\bibitem [{\citenamefont {Bisio}\ \emph {et~al.}(2010)\citenamefont {Bisio}, \citenamefont {Chiribella}, \citenamefont {D'Ariano}, \citenamefont {Facchini},\ and\ \citenamefont {Perinotti}}]{bisio2010optimal}%
  \BibitemOpen
  \bibfield  {author} {\bibinfo {author} {\bibfnamefont {A.}~\bibnamefont {Bisio}}, \bibinfo {author} {\bibfnamefont {G.}~\bibnamefont {Chiribella}}, \bibinfo {author} {\bibfnamefont {G.~M.}\ \bibnamefont {D'Ariano}}, \bibinfo {author} {\bibfnamefont {S.}~\bibnamefont {Facchini}},\ and\ \bibinfo {author} {\bibfnamefont {P.}~\bibnamefont {Perinotti}},\ }\bibfield  {title} {\bibinfo {title} {Optimal quantum learning of a unitary transformation},\ }\href {https://doi.org/10.1103/PhysRevA.81.032324} {\bibfield  {journal} {\bibinfo  {journal} {Phys. Rev. A}\ }\textbf {\bibinfo {volume} {81}},\ \bibinfo {pages} {032324} (\bibinfo {year} {2010})},\ \Eprint {https://arxiv.org/abs/0903.0543} {arXiv:0903.0543} \BibitemShut {NoStop}%
\bibitem [{\citenamefont {Sedl\'ak}\ \emph {et~al.}(2019)\citenamefont {Sedl\'ak}, \citenamefont {Bisio},\ and\ \citenamefont {Ziman}}]{sedlak2019optimal}%
  \BibitemOpen
  \bibfield  {author} {\bibinfo {author} {\bibfnamefont {M.}~\bibnamefont {Sedl\'ak}}, \bibinfo {author} {\bibfnamefont {A.}~\bibnamefont {Bisio}},\ and\ \bibinfo {author} {\bibfnamefont {M.}~\bibnamefont {Ziman}},\ }\bibfield  {title} {\bibinfo {title} {{Optimal Probabilistic Storage and Retrieval of Unitary Channels}},\ }\href {https://doi.org/10.1103/PhysRevLett.122.170502} {\bibfield  {journal} {\bibinfo  {journal} {Phys. Rev. Lett.}\ }\textbf {\bibinfo {volume} {122}},\ \bibinfo {pages} {170502} (\bibinfo {year} {2019})},\ \Eprint {https://arxiv.org/abs/1809.04552} {arXiv:1809.04552} \BibitemShut {NoStop}%
\bibitem [{\citenamefont {Sedl\'ak}\ and\ \citenamefont {Ziman}(2020)}]{sedlak2020probabilistic}%
  \BibitemOpen
  \bibfield  {author} {\bibinfo {author} {\bibfnamefont {M.}~\bibnamefont {Sedl\'ak}}\ and\ \bibinfo {author} {\bibfnamefont {M.}~\bibnamefont {Ziman}},\ }\bibfield  {title} {\bibinfo {title} {Probabilistic storage and retrieval of qubit phase gates},\ }\href {https://doi.org/10.1103/PhysRevA.102.032618} {\bibfield  {journal} {\bibinfo  {journal} {Phys. Rev. A}\ }\textbf {\bibinfo {volume} {102}},\ \bibinfo {pages} {032618} (\bibinfo {year} {2020})},\ \Eprint {https://arxiv.org/abs/2008.09555} {arXiv:2008.09555} \BibitemShut {NoStop}%
\bibitem [{\citenamefont {Bisio}\ \emph {et~al.}(2014)\citenamefont {Bisio}, \citenamefont {D'Ariano}, \citenamefont {Perinotti},\ and\ \citenamefont {Sedl{\'a}k}}]{bisio2014optimal}%
  \BibitemOpen
  \bibfield  {author} {\bibinfo {author} {\bibfnamefont {A.}~\bibnamefont {Bisio}}, \bibinfo {author} {\bibfnamefont {G.~M.}\ \bibnamefont {D'Ariano}}, \bibinfo {author} {\bibfnamefont {P.}~\bibnamefont {Perinotti}},\ and\ \bibinfo {author} {\bibfnamefont {M.}~\bibnamefont {Sedl{\'a}k}},\ }\bibfield  {title} {\bibinfo {title} {Optimal processing of reversible quantum channels},\ }\href {https://doi.org/10.1016/j.physleta.2014.04.042} {\bibfield  {journal} {\bibinfo  {journal} {Phys. Lett. A}\ }\textbf {\bibinfo {volume} {378}},\ \bibinfo {pages} {1797} (\bibinfo {year} {2014})},\ \Eprint {https://arxiv.org/abs/1308.3254} {arXiv:1308.3254} \BibitemShut {NoStop}%
\bibitem [{\citenamefont {D\"ur}\ \emph {et~al.}(2015)\citenamefont {D\"ur}, \citenamefont {Sekatski},\ and\ \citenamefont {Skotiniotis}}]{dur2015deterministic}%
  \BibitemOpen
  \bibfield  {author} {\bibinfo {author} {\bibfnamefont {W.}~\bibnamefont {D\"ur}}, \bibinfo {author} {\bibfnamefont {P.}~\bibnamefont {Sekatski}},\ and\ \bibinfo {author} {\bibfnamefont {M.}~\bibnamefont {Skotiniotis}},\ }\bibfield  {title} {\bibinfo {title} {{Deterministic Superreplication of One-Parameter Unitary Transformations}},\ }\href {https://doi.org/10.1103/PhysRevLett.114.120503} {\bibfield  {journal} {\bibinfo  {journal} {Phys. Rev. Lett.}\ }\textbf {\bibinfo {volume} {114}},\ \bibinfo {pages} {120503} (\bibinfo {year} {2015})},\ \Eprint {https://arxiv.org/abs/1410.6008} {arXiv:1410.6008} \BibitemShut {NoStop}%
\bibitem [{\citenamefont {Chiribella}\ \emph {et~al.}(2015)\citenamefont {Chiribella}, \citenamefont {Yang},\ and\ \citenamefont {Huang}}]{chiribella2015universal}%
  \BibitemOpen
  \bibfield  {author} {\bibinfo {author} {\bibfnamefont {G.}~\bibnamefont {Chiribella}}, \bibinfo {author} {\bibfnamefont {Y.}~\bibnamefont {Yang}},\ and\ \bibinfo {author} {\bibfnamefont {C.}~\bibnamefont {Huang}},\ }\bibfield  {title} {\bibinfo {title} {{Universal Superreplication of Unitary Gates}},\ }\href {https://doi.org/10.1103/PhysRevLett.114.120504} {\bibfield  {journal} {\bibinfo  {journal} {Phys. Rev. Lett.}\ }\textbf {\bibinfo {volume} {114}},\ \bibinfo {pages} {120504} (\bibinfo {year} {2015})},\ \Eprint {https://arxiv.org/abs/1412.1349} {arXiv:1412.1349} \BibitemShut {NoStop}%
\bibitem [{\citenamefont {Soleimanifar}\ and\ \citenamefont {Karimipour}(2016)}]{soleimanifar2016nogo}%
  \BibitemOpen
  \bibfield  {author} {\bibinfo {author} {\bibfnamefont {M.}~\bibnamefont {Soleimanifar}}\ and\ \bibinfo {author} {\bibfnamefont {V.}~\bibnamefont {Karimipour}},\ }\bibfield  {title} {\bibinfo {title} {No-go theorem for iterations of unknown quantum gates},\ }\href {https://doi.org/10.1103/PhysRevA.93.012344} {\bibfield  {journal} {\bibinfo  {journal} {Phys. Rev. A}\ }\textbf {\bibinfo {volume} {93}},\ \bibinfo {pages} {012344} (\bibinfo {year} {2016})},\ \Eprint {https://arxiv.org/abs/1510.06888} {arXiv:1510.06888} \BibitemShut {NoStop}%
\bibitem [{\citenamefont {Ebler}\ \emph {et~al.}(2023)\citenamefont {Ebler}, \citenamefont {Horodecki}, \citenamefont {Marciniak}, \citenamefont {M^^c5^^82ynik}, \citenamefont {Quintino},\ and\ \citenamefont {Studzi^^c5^^84ski}}]{ebler2023optimal}%
  \BibitemOpen
  \bibfield  {author} {\bibinfo {author} {\bibfnamefont {D.}~\bibnamefont {Ebler}}, \bibinfo {author} {\bibfnamefont {M.}~\bibnamefont {Horodecki}}, \bibinfo {author} {\bibfnamefont {M.}~\bibnamefont {Marciniak}}, \bibinfo {author} {\bibfnamefont {T.}~\bibnamefont {M^^c5^^82ynik}}, \bibinfo {author} {\bibfnamefont {M.~T.}\ \bibnamefont {Quintino}},\ and\ \bibinfo {author} {\bibfnamefont {M.}~\bibnamefont {Studzi^^c5^^84ski}},\ }\bibfield  {title} {\bibinfo {title} {{Optimal Universal Quantum Circuits for Unitary Complex Conjugation}},\ }\href {https://doi.org/10.1109/TIT.2023.3263771} {\bibfield  {journal} {\bibinfo  {journal} {IEEE Trans. Inf. Theory}\ }\textbf {\bibinfo {volume} {69}},\ \bibinfo {pages} {5069} (\bibinfo {year} {2023})},\ \Eprint {https://arxiv.org/abs/2206.00107} {arXiv:2206.00107} \BibitemShut {NoStop}%
\bibitem [{\citenamefont {Ara{\'u}jo}\ \emph {et~al.}(2014)\citenamefont {Ara{\'u}jo}, \citenamefont {Feix}, \citenamefont {Costa},\ and\ \citenamefont {Brukner}}]{araujo2014quantum}%
  \BibitemOpen
  \bibfield  {author} {\bibinfo {author} {\bibfnamefont {M.}~\bibnamefont {Ara{\'u}jo}}, \bibinfo {author} {\bibfnamefont {A.}~\bibnamefont {Feix}}, \bibinfo {author} {\bibfnamefont {F.}~\bibnamefont {Costa}},\ and\ \bibinfo {author} {\bibfnamefont {{\v{C}}.}~\bibnamefont {Brukner}},\ }\bibfield  {title} {\bibinfo {title} {{Quantum circuits cannot control unknown operations}},\ }\href {https://doi.org/10.1088/1367-2630/16/9/093026} {\bibfield  {journal} {\bibinfo  {journal} {New J. Phys.}\ }\textbf {\bibinfo {volume} {16}},\ \bibinfo {pages} {093026} (\bibinfo {year} {2014})},\ \Eprint {https://arxiv.org/abs/1309.7976} {arXiv:1309.7976} \BibitemShut {NoStop}%
\bibitem [{\citenamefont {Bisio}\ \emph {et~al.}(2016)\citenamefont {Bisio}, \citenamefont {Dall'Arno},\ and\ \citenamefont {Perinotti}}]{bisio2016quantum}%
  \BibitemOpen
  \bibfield  {author} {\bibinfo {author} {\bibfnamefont {A.}~\bibnamefont {Bisio}}, \bibinfo {author} {\bibfnamefont {M.}~\bibnamefont {Dall'Arno}},\ and\ \bibinfo {author} {\bibfnamefont {P.}~\bibnamefont {Perinotti}},\ }\bibfield  {title} {\bibinfo {title} {Quantum conditional operations},\ }\href {https://doi.org/10.1103/PhysRevA.94.022340} {\bibfield  {journal} {\bibinfo  {journal} {Phys. Rev. A}\ }\textbf {\bibinfo {volume} {94}},\ \bibinfo {pages} {022340} (\bibinfo {year} {2016})},\ \Eprint {https://arxiv.org/abs/1509.01062} {arXiv:1509.01062} \BibitemShut {NoStop}%
\bibitem [{\citenamefont {Dong}\ \emph {et~al.}(2021)\citenamefont {Dong}, \citenamefont {Quintino}, \citenamefont {Soeda},\ and\ \citenamefont {Murao}}]{dong2021success}%
  \BibitemOpen
  \bibfield  {author} {\bibinfo {author} {\bibfnamefont {Q.}~\bibnamefont {Dong}}, \bibinfo {author} {\bibfnamefont {M.~T.}\ \bibnamefont {Quintino}}, \bibinfo {author} {\bibfnamefont {A.}~\bibnamefont {Soeda}},\ and\ \bibinfo {author} {\bibfnamefont {M.}~\bibnamefont {Murao}},\ }\bibfield  {title} {\bibinfo {title} {{Success-or-Draw: A Strategy Allowing Repeat-Until-Success in Quantum Computation}},\ }\href {https://doi.org/10.1103/PhysRevLett.126.150504} {\bibfield  {journal} {\bibinfo  {journal} {Phys. Rev. Lett.}\ }\textbf {\bibinfo {volume} {126}},\ \bibinfo {pages} {150504} (\bibinfo {year} {2021})},\ \Eprint {https://arxiv.org/abs/2011.01055} {arXiv:2011.01055} \BibitemShut {NoStop}%
\bibitem [{\citenamefont {Miyazaki}\ \emph {et~al.}(2019)\citenamefont {Miyazaki}, \citenamefont {Soeda},\ and\ \citenamefont {Murao}}]{miyazaki2019complex}%
  \BibitemOpen
  \bibfield  {author} {\bibinfo {author} {\bibfnamefont {J.}~\bibnamefont {Miyazaki}}, \bibinfo {author} {\bibfnamefont {A.}~\bibnamefont {Soeda}},\ and\ \bibinfo {author} {\bibfnamefont {M.}~\bibnamefont {Murao}},\ }\bibfield  {title} {\bibinfo {title} {{Complex conjugation supermap of unitary quantum maps and its universal implementation protocol}},\ }\href {https://doi.org/10.1103/PhysRevResearch.1.013007} {\bibfield  {journal} {\bibinfo  {journal} {Phys. Rev. Res.}\ }\textbf {\bibinfo {volume} {1}},\ \bibinfo {pages} {013007} (\bibinfo {year} {2019})},\ \Eprint {https://arxiv.org/abs/1706.03481} {arXiv:1706.03481} \BibitemShut {NoStop}%
\bibitem [{\citenamefont {Grinko}\ and\ \citenamefont {Ozols}(2024)}]{grinko2024linear}%
  \BibitemOpen
  \bibfield  {author} {\bibinfo {author} {\bibfnamefont {D.}~\bibnamefont {Grinko}}\ and\ \bibinfo {author} {\bibfnamefont {M.}~\bibnamefont {Ozols}},\ }\bibfield  {title} {\bibinfo {title} {Linear programming with unitary-equivariant constraints},\ }\href {https://doi.org/10.1007/s00220-024-05108-1} {\bibfield  {journal} {\bibinfo  {journal} {Commun. Math. Phys.}\ }\textbf {\bibinfo {volume} {405}},\ \bibinfo {pages} {278} (\bibinfo {year} {2024})},\ \Eprint {https://arxiv.org/abs/2207.05713} {arXiv:2207.05713} \BibitemShut {NoStop}%
\bibitem [{\citenamefont {Chiribella}\ \emph {et~al.}(2008{\natexlab{b}})\citenamefont {Chiribella}, \citenamefont {D’Ariano},\ and\ \citenamefont {Perinotti}}]{chiribella2008quantum}%
  \BibitemOpen
  \bibfield  {author} {\bibinfo {author} {\bibfnamefont {G.}~\bibnamefont {Chiribella}}, \bibinfo {author} {\bibfnamefont {G.~M.}\ \bibnamefont {D’Ariano}},\ and\ \bibinfo {author} {\bibfnamefont {P.}~\bibnamefont {Perinotti}},\ }\bibfield  {title} {\bibinfo {title} {{Quantum circuit architecture}},\ }\href {https://doi.org/10.1103/PhysRevLett.101.060401} {\bibfield  {journal} {\bibinfo  {journal} {Phys. Rev. Lett.}\ }\textbf {\bibinfo {volume} {101}},\ \bibinfo {pages} {060401} (\bibinfo {year} {2008}{\natexlab{b}})},\ \Eprint {https://arxiv.org/abs/0712.1325} {arXiv:0712.1325} \BibitemShut {NoStop}%
\bibitem [{\citenamefont {Montanaro}\ and\ \citenamefont {Wolf}(2016)}]{montanaro2013survey}%
  \BibitemOpen
  \bibfield  {author} {\bibinfo {author} {\bibfnamefont {A.}~\bibnamefont {Montanaro}}\ and\ \bibinfo {author} {\bibfnamefont {R.~d.}\ \bibnamefont {Wolf}},\ }\href {https://doi.org/10.4086/toc.gs.2016.007} {\emph {\bibinfo {title} {{A Survey of Quantum Property Testing}}}},\ \bibinfo {series} {Graduate Surveys}\ No.~\bibinfo {number} {7}\ (\bibinfo  {publisher} {Theory of Computing Library},\ \bibinfo {year} {2016})\ pp.\ \bibinfo {pages} {1--81},\ \Eprint {https://arxiv.org/abs/1310.2035} {arXiv:1310.2035} \BibitemShut {NoStop}%
\bibitem [{\citenamefont {Harrow}\ \emph {et~al.}(2017)\citenamefont {Harrow}, \citenamefont {Lin},\ and\ \citenamefont {Montanaro}}]{harrow2017sequential}%
  \BibitemOpen
  \bibfield  {author} {\bibinfo {author} {\bibfnamefont {A.~W.}\ \bibnamefont {Harrow}}, \bibinfo {author} {\bibfnamefont {C.~Y.-Y.}\ \bibnamefont {Lin}},\ and\ \bibinfo {author} {\bibfnamefont {A.}~\bibnamefont {Montanaro}},\ }\bibfield  {title} {\bibinfo {title} {Sequential measurements, disturbance and property testing},\ }in\ \href {https://doi.org/10.1137/1.9781611974782.105} {\emph {\bibinfo {booktitle} {Proceedings of the Twenty-Eighth Annual ACM-SIAM Symposium on Discrete Algorithms}}}\ (\bibinfo {organization} {SIAM},\ \bibinfo {year} {2017})\ pp.\ \bibinfo {pages} {1598--1611},\ \Eprint {https://arxiv.org/abs/1607.03236} {arXiv:1607.03236} \BibitemShut {NoStop}%
\bibitem [{\citenamefont {Gavorov{\'a}}\ \emph {et~al.}(2024)\citenamefont {Gavorov{\'a}}, \citenamefont {Seidel},\ and\ \citenamefont {Touati}}]{gavorova2024topological}%
  \BibitemOpen
  \bibfield  {author} {\bibinfo {author} {\bibfnamefont {Z.}~\bibnamefont {Gavorov{\'a}}}, \bibinfo {author} {\bibfnamefont {M.}~\bibnamefont {Seidel}},\ and\ \bibinfo {author} {\bibfnamefont {Y.}~\bibnamefont {Touati}},\ }\bibfield  {title} {\bibinfo {title} {Topological obstructions to quantum computation with unitary oracles},\ }\href {https://doi.org/10.1103/PhysRevA.109.032625} {\bibfield  {journal} {\bibinfo  {journal} {Phys. Rev. A}\ }\textbf {\bibinfo {volume} {109}},\ \bibinfo {pages} {032625} (\bibinfo {year} {2024})},\ \Eprint {https://arxiv.org/abs/2011.10031} {arXiv:2011.10031} \BibitemShut {NoStop}%
\bibitem [{\citenamefont {Choi}(1975)}]{choi1975completely}%
  \BibitemOpen
  \bibfield  {author} {\bibinfo {author} {\bibfnamefont {M.-D.}\ \bibnamefont {Choi}},\ }\bibfield  {title} {\bibinfo {title} {{Completely positive linear maps on complex matrices}},\ }\href {https://doi.org/10.1016/0024-3795(75)90075-0} {\bibfield  {journal} {\bibinfo  {journal} {Linear Algebra Appl.}\ }\textbf {\bibinfo {volume} {10}},\ \bibinfo {pages} {285} (\bibinfo {year} {1975})}\BibitemShut {NoStop}%
\bibitem [{\citenamefont {Jamio{\l}kowski}(1972)}]{jamiolkowski1972linear}%
  \BibitemOpen
  \bibfield  {author} {\bibinfo {author} {\bibfnamefont {A.}~\bibnamefont {Jamio{\l}kowski}},\ }\bibfield  {title} {\bibinfo {title} {{Linear transformations which preserve trace and positive semidefiniteness of operators}},\ }\href {https://doi.org/10.1016/0034-4877(72)90011-0} {\bibfield  {journal} {\bibinfo  {journal} {Rep. Math. Phys.}\ }\textbf {\bibinfo {volume} {3}},\ \bibinfo {pages} {275} (\bibinfo {year} {1972})}\BibitemShut {NoStop}%
\bibitem [{sup()}]{supple}%
  \BibitemOpen
  \href@noop {} {}\bibinfo {note} {See Supplemental Material for details.}\BibitemShut {Stop}%
\bibitem [{\citenamefont {Fang}\ \emph {et~al.}(2006)\citenamefont {Fang}, \citenamefont {Fenner}, \citenamefont {Green}, \citenamefont {Homer},\ and\ \citenamefont {Zhang}}]{fang2003quantum}%
  \BibitemOpen
  \bibfield  {author} {\bibinfo {author} {\bibfnamefont {M.}~\bibnamefont {Fang}}, \bibinfo {author} {\bibfnamefont {S.}~\bibnamefont {Fenner}}, \bibinfo {author} {\bibfnamefont {F.}~\bibnamefont {Green}}, \bibinfo {author} {\bibfnamefont {S.}~\bibnamefont {Homer}},\ and\ \bibinfo {author} {\bibfnamefont {Y.}~\bibnamefont {Zhang}},\ }\bibfield  {title} {\bibinfo {title} {Quantum lower bounds for fanout},\ }\href {https://doi.org/10.5555/2011679.2011682} {\bibfield  {journal} {\bibinfo  {journal} {Quantum Info. Comput.}\ }\textbf {\bibinfo {volume} {6}},\ \bibinfo {pages} {46} (\bibinfo {year} {2006})},\ \Eprint {https://arxiv.org/abs/quant-ph/0312208} {arXiv:quant-ph/0312208} \BibitemShut {NoStop}%
\bibitem [{\citenamefont {Bavaresco}\ \emph {et~al.}(2025)\citenamefont {Bavaresco}, \citenamefont {Kristj{\'a}nsson}, \citenamefont {Murao}, \citenamefont {Odake}, \citenamefont {Quintino}, \citenamefont {Taranto},\ and\ \citenamefont {Yoshida}}]{bavaresco2025simulating}%
  \BibitemOpen
  \bibfield  {author} {\bibinfo {author} {\bibfnamefont {J.}~\bibnamefont {Bavaresco}}, \bibinfo {author} {\bibfnamefont {H.}~\bibnamefont {Kristj{\'a}nsson}}, \bibinfo {author} {\bibfnamefont {M.}~\bibnamefont {Murao}}, \bibinfo {author} {\bibfnamefont {T.}~\bibnamefont {Odake}}, \bibinfo {author} {\bibfnamefont {M.~T.}\ \bibnamefont {Quintino}}, \bibinfo {author} {\bibfnamefont {P.}~\bibnamefont {Taranto}},\ and\ \bibinfo {author} {\bibfnamefont {S.}~\bibnamefont {Yoshida}},\ }\bibfield  {title} {\bibinfo {title} {Simulating the quantum switch with quantum circuits is computationally hard},\ }\href {https://doi.org/10.1038/s41467-025-64996-6} {\bibfield  {journal} {\bibinfo  {journal} {Nat. Commun.}\ }\textbf {\bibinfo {volume} {16}},\ \bibinfo {pages} {10216} (\bibinfo {year} {2025})},\ \Eprint {https://arxiv.org/abs/2409.18202} {arXiv:2409.18202} \BibitemShut {NoStop}%
\bibitem [{\citenamefont {Chen}\ \emph {et~al.}(2025)\citenamefont {Chen}, \citenamefont {Yu},\ and\ \citenamefont {Zhang}}]{chen2025tight}%
  \BibitemOpen
  \bibfield  {author} {\bibinfo {author} {\bibfnamefont {K.}~\bibnamefont {Chen}}, \bibinfo {author} {\bibfnamefont {N.}~\bibnamefont {Yu}},\ and\ \bibinfo {author} {\bibfnamefont {Z.}~\bibnamefont {Zhang}},\ }\bibfield  {title} {\bibinfo {title} {{Tight Bound for Quantum Unitary Time-Reversal}},\ }\Eprint {https://arxiv.org/abs/2507.05736} {arXiv:2507.05736}  (\bibinfo {year} {2025})\BibitemShut {NoStop}%
\end{thebibliography}%

\clearpage

\renewcommand{\thetable}{S\arabic{table}}
\renewcommand{\thefigure}{S\arabic{figure}}
\renewcommand{\theTheorem}{S\arabic{Theorem}}
\renewcommand{\thelem}{S\arabic{lem}}
\setcounter{equation}{0}
\setcounter{table}{0}
\setcounter{figure}{0}

\appendix
\onecolumngrid

\begin{center}
    \textbf{\large Supplemental Material for ``Analytical Lower Bound on Query Complexity for Transformations of Unknown Unitary Operations''}
\end{center}
\twocolumngrid

This Supplemental Material for ``Analytical Lower Bound on Query Complexity for Transformations of Unknown Unitary Operations'' is organized as follows.
Appendix~\ref{app::theo1} provides the proof of Thm.~\ref{th::main_sdp} and Cor.~\ref{cor:logdepth} in the main text, which show a lower bound of query complexity of a differentiable function $f$ in terms of the semidefinite programming (SDP).
Appendix~\ref{app:prooftable} provides the proof of Cor.~\ref{cor::inversion-transposistion-conjugation} in the main text, which shows a lower bound of query complexity of unitary inversion, transposition, and complex conjugation.
In addition, we also analyze unitary iteration.
Appendix~\ref{app::subgroup} shows an SDP providing a lower bound of query complexity for the partially-known setting, where the input unitary operation is promised to be within a subgroup of $\SU(d)$.
Appendix~\ref{app::prob} shows an SDP providing a lower bound of query complexity that provides a tradeoff between query complexity and success probability.
Appendix~\ref{app::duals} derives dual SDPs for the SDPs presented in this work.

\section{Proof of Thm.~\ref{th::main_sdp} and Cor.~\ref{cor:logdepth} in the main text}\label{app::theo1}

\subsection{Derivation of equations}
Before proving the theorem, we first derive the key equations used in the proof.
\begin{figure}[H]
    \centering
    \includegraphics[width=1.0\linewidth]{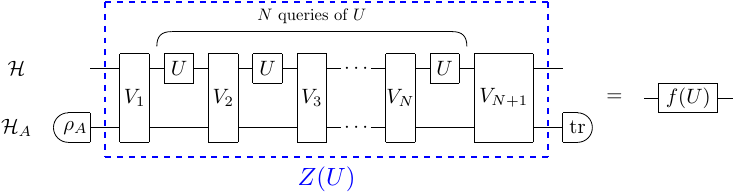}
    \caption{Quantum circuit implementing $f(U)$ using $N$ queries of an unknown unitary operation $U$. The upper and lower lines represent the main system $\mcH$ and the auxiliary system $\mcH_A$, respectively, $\rho_A$ is a quantum state of the auxiliary system, and $V_1,\ldots ,V_{N+1}$ are unitary operators on the composite system. $Z(U)$ is the unitary operation corresponding to the circuit excluding $\rho_A$ and tracing out.}
    \label{fig::lowerbound_comb}
\end{figure}

In the quantum circuit model of quantum computation, a transformation $f: \SU(d) \rightarrow \SU(d)$ of an unknown unitary operation $U \in  \SU(d)$ on a system $\mcH$ can always be represented by a fixed-order quantum circuit (quantum comb) with $N$-slots for querying $U$ shown in the left-hand side of Fig.~\ref{fig::lowerbound_comb}, as shown in \cite{chiribella2008quantum}, where $\rho_A$ is a quantum state of the auxiliary system $\mcH_A$ and $V_1, \ldots, V_{N+1}$ are unitary operators on $\mcH \otimes \mcH_A$.
The state $\rho_A$ in Fig.~\ref{fig::lowerbound_comb} can be taken as $\ketbra{0}{0}$ where $\ket{0}$ is one of the basis state in the computational basis $\{\ket{j}\}_j$ of $\mcH_A$ without loss of generality.
Defining a unitary operator $Z(U)$ on $\mcH\otimes \mcH_A$ by
\begin{align}
    Z(U) \coloneqq V_{N+1}\left(\prod_{j=1}^{N}(U\otimes I)V_j\right),
\end{align}
the state $Z(U)(\ket{\psi}\otimes \ket{0})$ for $\ket{\psi}\in \mcH$ has to satisfy
\begin{align}
    \Tr_{\mcH_A}[Z(U)(\ketbra{\psi}\otimes \ketbra{0})Z(U)^\dagger] = f(U)\ketbra{\psi}{\psi}f(U)^{\dagger}.
\end{align}
Therefore, we obtain
\begin{align}
\label{eq::one_side}
Z(U) \ket{\psi}\otimes \ket{0} &= 
V_{N+1}\left(\prod_{j=1}^{N}(U\otimes I)V_j\right)[\ket{\psi}\otimes \ket{0}]
\nonumber \\
&=f(U)\ket{\psi}\otimes \ket{\phi (U)},
\end{align}
where $\ket{\phi (U)}\in \mcH_A$ is a $U$-dependent state. Note that $\ket{\phi (U)}$ is independent of the input state $\ket{\psi}$ of $\mcH$, since if $\ket{\phi(U)}$ for the two input states $\ket{\psi}$ and $\ket{\psi'}$ are different under the same $U$, then the output state is no longer a product state if the input is taken proportional to $\ket{\psi}+\ket{\psi'}$, which contradicts Eq.~(\ref{eq::one_side}).

Finally, we add additional $U_0$-dependent gates $f(U_0)^{-1}\otimes W_{U_0}$ at the end of the quantum circuit, where $W_{U_0}$ is a unitary operation satisfying $W_{U_0}\ket{\phi (U_0)}=\ket{0}$, so that Eq.~(\ref{eq::one_side}) can be rewritten using $\tilde{V}_{N+1}(U_0)\coloneqq(f(U_0)^{-1}\otimes W_{U_0})V_{N+1}$ as
\begin{align}\label{eq::deepest_preprocessing}
    &\tilde{V}_{N+1}(U_0)
    \left(\prod_{j=1}^{N}(U\otimes I)V_j\right)[\ket{\psi}\otimes \ket{0}]
    \nonumber\\
    =&f(U_0)^{-1}f(U)\ket{\psi}\otimes W_{U_0}\ket{\phi (U)}.
\end{align}
Although Eq.~(\ref{eq::deepest_preprocessing}) looks more complicated than Eq.~(\ref{eq::one_side}), Eq.~(\ref{eq::deepest_preprocessing}) has a property that simplifies the expression in the neighborhood of $U=U_0$. 
In particular, taking $U=U_0$ gives
\begin{align}\label{eq::U=U_0}
    &\tilde{V}_{N+1}(U_0)\left(\prod_{j=1}^{N}(U_0\otimes I)V_j\right)[\ket{\psi}\otimes \ket{0}]=\ket{\psi}\otimes \ket{0}.
\end{align}
Moreover, by taking $U=(I+i\epsilon H+O(\epsilon^2))U_0$ ($H$: an Hermitian operator, $\epsilon\ll 1$) and considering the first-order $\epsilon$ terms, we can obtain another equation, which is used in the proof of Thm.~\ref{th::main_sdp} together with Eq.~(\ref{eq::U=U_0}).
Equation~(\ref{eq::U=U_0}) is illustrated by the quantum circuit shown in Fig.~\ref{fig::lowerbound_simplified}.

\begin{figure}[H]
    \centering
    \includegraphics[width=1.0\linewidth]{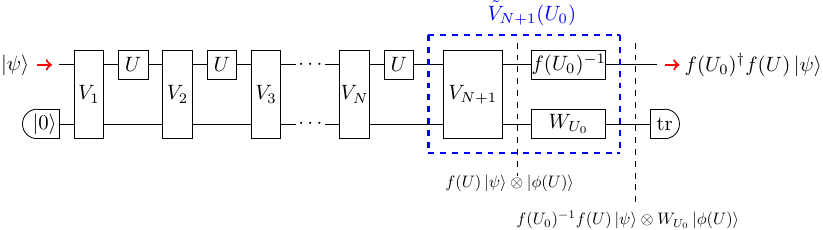}
    \caption{Quantum circuit in Fig.~\ref{fig::lowerbound_comb} is transformed to this circuit in order to simplify the expression in the neighborhood of $U=U_0$.}
    \label{fig::lowerbound_simplified}
\end{figure}

\subsection{Lemmas for proving the Thm.~\ref{th::main_sdp}}
We now prove lemmas used in the proof of Thm.~\ref{th::main_sdp}.
\begin{lem}\label{le::Vlefrig}
Let us define $V^{(s,{\rm left})}(U_0)$ and $V^{(s,{\rm right})}(U_0)$ ($s\in \{1,\ldots ,N\}$) as
\begin{figure}[H]
    \centering
    \includegraphics[width=0.8\linewidth]{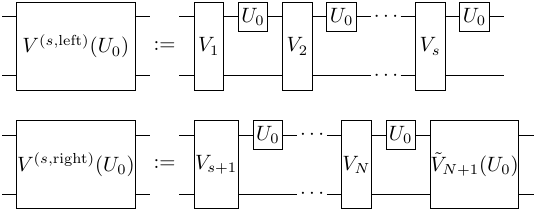}
\end{figure}
\noindent so that for all $s\in \{1,\ldots ,N\}$, we have
\begin{align}
    V^{(s,{\rm right})}(U_0)V^{(s,{\rm left})}(U_0)=\tilde{V}_{N+1}(U_0)\left(\prod_{j=1}^{N}(U_0\otimes I)V_j\right).
\end{align}
Let us also define $M^{(s,{\rm left})}_{j,k}(U_0)$ and $ M^{(s,{\rm right})}_{j,k}(U_0)$ as the $\ketbra{j}{k}$-auxiliary-block of $V^{(s,{\rm left})}(U_0)$ and $V^{(s,{\rm right})}(U_0)$, respectively, namely,
\begin{align}\label{eq::main_M_def}
    \begin{split}
    &V^{(s,{\rm left})}(U_0)\eqqcolon\sum_{j,k}
    M^{(s,{\rm left})}_{j,k}(U_0)\otimes \ketbra{j}{k}\\
    &V^{(s,{\rm right})}(U_0)\eqqcolon\sum_{j,k}
    M^{(s,{\rm right})}_{j,k}(U_0)\otimes \ketbra{j}{k}.
    \end{split}
\end{align}
Then, we obtain
\begin{align}\label{eq::righttoleft}
    M^{(s,{\rm right})}_{0,j}(U_0)^{\dagger}
    =
    M^{(s,{\rm left})}_{j,0}(U_0).
\end{align}
\end{lem}

\textbf{Proof: }
Equation~(\ref{eq::U=U_0}) can be rewritten as
\begin{align}
    I
    &=(I\otimes \bra{0})\left[\tilde{V}_{N+1}(U_0)\left(\prod_{j=1}^{N}(U_0\otimes I)V_j\right)\right](I\otimes \ket{0})\nonumber\\
    &=\left(\sum_{\ell}
    M^{(s,{\rm right})}_{0,\ell}(U_0)\otimes \bra{\ell}
    \right)
    \left(\sum_{j}
    M^{(s,{\rm left})}_{j,0}(U_0)\otimes \ket{j}
    \right)\nonumber\\
    &=
    \sum_j M^{(s,{\rm right})}_{0,j}(U_0)
    M^{(s,{\rm left})}_{j,0}(U_0)\quad (s\in \{1,\ldots ,N\}) ,
\end{align}
thus, by taking the trace, we obtain
\begin{align}\label{eq::leftright2}
    \sum_j {\rm tr}\left(M^{(s,{\rm right})}_{0,j}(U_0)
    M^{(s,{\rm left})}_{j,0}(U_0)\right)
    ={\rm tr}I=d.
\end{align}
Note that the left-hand side of Eq.~(\ref{eq::leftright2}) can be seen as an inner product on the set of linear operators. Namely, by defining the inner product $(\{A_j\}_j, \{B_k\}_k)$ of two sets $\{A_j\}_j$ and $\{B_k\}_k$ (indices $j,k$ are taken from the same sets) of linear operators on $\mcL(\mcH)$ as
\begin{align}
    (\{A_j\}_j, \{B_k\}_k)\coloneqq\sum_j{\rm tr} A_j^{\dagger}B_j,
\end{align}
which can be seen as a straightforward extension of the Hilbert-Schmidt inner product, Eq.~(\ref{eq::leftright2}) can be rewritten as
\begin{align}\label{eq::leftright}
    \left(\{M^{(s,{\rm right})}_{0,\ell}(U_0)^{\dagger}\}_{\ell}, \{M^{(s,{\rm left})}_{j,0}(U_0)\}_j\right)
    = d.
\end{align}
On the other hand, similar equations involving the inner product can be obtained from the unitary operators $V^{(s,{\rm left})}(U_0)$ and $V^{(s,{\rm right})}(U_0)$, namely, 
\begin{align}
    I
    &=(I\otimes \bra{0})[V^{(s,{\rm left})}(U_0)^{\dagger}V^{(s,{\rm left})}(U_0)](I\otimes \ket{0})
    \nonumber\\
    &=\sum_j M^{(s,{\rm left})}_{j,0}(U_0)^{\dagger}M^{(s,{\rm left})}_{j,0}(U_0),\\
    I
    &=(I\otimes \bra{0})[V^{(s,{\rm right})}(U_0)V^{(s,{\rm right})}(U_0)^{\dagger}](I\otimes \ket{0})\nonumber\\
    &=\sum_{\ell}M^{(s,{\rm right})}_{0,\ell}(U_0)M^{(s,{\rm right})}_{0,\ell}(U_0)^{\dagger},
\end{align}
thus
\begin{align}
\label{eq::leftleft}
    \left(\{M^{(s,{\rm left})}_{j,0}(U_0)\}_{j}, \{M^{(s,{\rm left})}_{j,0}(U_0)\}_j\right)
    &= d,\\
\label{eq::rightright}
    \left(\{M^{(s,{\rm right})}_{0,\ell}(U_0)^{\dagger}\}_{\ell}, \{M^{(s,{\rm right})}_{0,\ell}(U_0)^{\dagger}\}_{\ell}\right)
    &= d
\end{align}
hold.
By combining Eqs.~(\ref{eq::leftright}), (\ref{eq::leftleft}), and (\ref{eq::rightright}), we obtain
\begin{align}
    &\left|\left(\{M^{(s,{\rm right})}_{0,\ell}(U_0)^{\dagger}\}_{\ell}, \{M^{(s,{\rm left})}_{j,0}(U_0)\}_j\right)\right|^2
    \nonumber\\
    =& 
    \left(\{M^{(s,{\rm right})}_{0,\ell}(U_0)^{\dagger}\}_{\ell}, \{M^{(s,{\rm right})}_{0,\ell}(U_0)^{\dagger}\}_{\ell}\right)
    \cdot
    \nonumber\\
    &\left(\{M^{(s,{\rm left})}_{j,0}(U_0)\}_{j}, \{M^{(s,{\rm left})}_{j,0}(U_0)\}_j\right).
\end{align}
Since the equality of the Cauchy-Schwarz inequality only holds when 
\begin{align}
    M^{(s,{\rm right})}_{0,j}(U_0)^{\dagger}
    \propto 
    M^{(s,{\rm left})}_{j,0}(U_0)
\end{align}
and $\{M^{(s,{\rm right})}_{0,\ell}(U_0)^{\dagger}\}_{\ell}$ and $\{M^{(s,{\rm left})}_{j,0}(U_0)\}_{j}$ have the same norm, we obtain
\begin{align}
    M^{(s,{\rm right})}_{0,j}(U_0)^{\dagger}
    =
    M^{(s,{\rm left})}_{j,0}(U_0).
\end{align} \qed

By using Lem.~\ref{le::Vlefrig}, we obtain the following Lemma.

\begin{lem}\label{le::CP_comb}
    The map $\mcE_{U_0}$ defined as
    \begin{align}\label{eq::eU0_def_orig}
        \mcE_{U_0}(H)=
        \sum_{s=1}^{N}
        \sum_j
        (M^{(s,{\rm left})}_{j,0}(U_0))^{\dagger} H (M^{(s,{\rm left})}_{j,0}(U_0))
    \end{align}
    satisfies
    \begin{align}\label{eq::comb_vs_f}
    \begin{cases}
        \mcE_{U_0}(I)=NI\\
        \mcE_{U_0}(H)=g_{U_0}(H)+\alpha_{U_0}(H)I&(H\in \su (d))
    \end{cases}
    \end{align}
    for a linear map $\alpha_{U_0}:\su(d)\to \mathbb{R}$.
\end{lem}
Note that $\mcE_{U_0}$ is completely positive. From this property, the problem of finding the lower bound of $N$ can be reduced to the SDP.

\textbf{Proof: }
By substituting $U=e^{i\epsilon H}U_0$ ($H\in \su (d)$) to Eq.~(\ref{eq::deepest_preprocessing}) and taking the derivative by $\epsilon$ around $\epsilon=0$, we obtain
\begin{align}
\label{eq:a22}
    &\sum_{s=1}^N
    V^{(s, {\rm right})}(U_0)
    (iH\otimes I)
    V^{(s, {\rm left})}(U_0)
    [\ket{\psi}\otimes \ket{0}]\nonumber\\
    =&\left.\frac{{\rm d}}{{\rm d}\epsilon}\right|_{\epsilon=0}
    \left[
    \tilde{V}_{N+1}(U_0)
    \left(\prod_{j=1}^{N}(e^{i\epsilon H}U_0\otimes I)V_j\right)[\ket{\psi}\otimes \ket{0}]
    \right]
    \nonumber\\
    =&
    \left.\frac{{\rm d}}{{\rm d}\epsilon}\right|_{\epsilon=0}
    \left[
    f(U_0)^{-1}f(e^{i\epsilon H}U_0)\ket{\psi}\otimes W_{U_0}\ket{\phi (e^{i\epsilon H}U_0)}
    \right]
    \nonumber\\
    =&ig_{U_0}(H)\ket{\psi}\otimes \ket{0}
    +
    \ket{\psi}\otimes \left.\frac{{\rm d}}{{\rm d}\epsilon}\right|_{\epsilon=0}W_{U_0}\ket{\phi (e^{i\epsilon H}U_0)}
    .
\end{align}
The first equality is obtained using the product rule (Leibniz rule), namely the derivative is equal to the sum of the terms where the derivative is applied on the $s$-th $e^{i\epsilon H}U_0$ ($s\in \{1,\ldots ,N\}$) and $\epsilon\to 0$ (namely $e^{i\epsilon H}U_0\to U_0$) for the rest of terms.
Defining $\ket{\Phi (U_0,H)}\coloneqq \left.\frac{{\rm d}}{{\rm d}\epsilon}\right|_{\epsilon=0} W_{U_0}\ket{\phi (e^{i\epsilon H}U_0)}$ which is linear in $H$, we obtain
\begin{align}
    0
    &=\left.\frac{{\rm d}}{{\rm d}\epsilon}\right|_{\epsilon=0}
    \bra{\phi (e^{i\epsilon H}U_0)}
    W_{U_0}^{\dagger} W_{U_0}
    \ket{\phi (e^{i\epsilon H}U_0)}
    \nonumber\\
    &=\braket{0}{\Phi (U_0,H)}+\braket{\Phi (U_0,H)}{0}.
\end{align}
Thus, $\braket{0}{\Phi (U_0,H)}$ can be expressed as
\begin{align}
    \braket{0}{\Phi (U_0,H)}\coloneqq i\alpha_{U_0}(H)
\end{align}
using a linear map $\alpha_{U_0}:\mcL(\mcH)\to \mathbb{R}$.

Therefore, by applying $I\otimes \bra{0}$ from left on Eq.~\eqref{eq:a22}, we obtain
\begin{align}
    &\sum_{s=1}^N\sum_{j,\ell}\left(M^{(s,{\rm right})}_{0,\ell}(U_0)\otimes \bra{\ell}\right)(iH\otimes I)\left(M^{(s,{\rm left})}_{j,0}(U_0)\otimes \ket{j}\right)\nonumber\\
    &=ig_{U_0}(H)+i\alpha_{U_0}(U_0)I,
\end{align}
i.e.,
\begin{align}
    \mcE_{U_0}(H)
    &=\sum_{s=1}^N\sum_{j} 
    M^{(s,{\rm left})}_{j,0}(U_0)^{\dagger}HM^{(s,{\rm left})}_{j,0}(U_0)\nonumber\\
    &=\sum_{s=1}^N\sum_{j} 
    M^{(s,{\rm right})}_{0,j}(U_0)HM^{(s,{\rm left})}_{j,0}(U_0)\nonumber\\
    &= g_{U_0}(H)+\alpha_{U_0}(U_0)I,
\end{align}
where we use Eq.~\eqref{eq::righttoleft} in the second equality.
Since
\begin{align}
    \mcE_{U_0}(I)&=\sum_{s=1}^N\sum_{j} 
    M^{(s,{\rm left})}_{j,0}(U_0)^{\dagger} M^{(s,{\rm left})}_{j,0}(U_0)\nonumber\\
    &= \sum_{s=1}^{N}(I\otimes \bra{0})V^{(s,{\rm left})\dagger}(U_0) V^{(s,{\rm left})}(U_0)(I\otimes \ket{0})\nonumber\\
    &= NI
\end{align}
holds, we obtain Lem.~\ref{le::CP_comb}.
\qed

\subsection{Proof of Thm.~\ref{th::main_sdp}}
From Lem.~\ref{le::CP_comb}, we can show a lower bound on the query complexity of $f$ given by the following optimization problem:
\begin{align}
\begin{split}
    &\min N\\
    \text{s.t. }& \mcE_{U_0} \text{ is CP}, \alpha_{U_0}: \su(d) \to \su(d) \text{ is linear},\\
    & \mcE_{U_0}(I)=NI,\\
    & \mcE_{U_0}(H) = g_{U_0}(H) + \alpha_{U_0}(H)I \quad \forall H\in\su(d).
\end{split}
\label{eq:optimization}
\end{align}
By defining the Choi operator of $g_{U_0}$ by
\begin{align}
    J_{g_{U_0}} \coloneqq \sum_{j=1}^{d^2-1} B_j^* \otimes g_{U_0}(B_j),
\end{align}
for any orthonormal basis $\{B_j\}_j$ of $\su (d)$ and defining $\beta_{U_0}$ satisfying
\begin{align}
    &\tr(\beta_{U_0}^T I)=N
    \nonumber\\
    &\tr(\beta_{U_0}^T H)=\alpha_{U_0}(H)\quad (\forall H\in \su (d)),
\end{align}
the Choi operator of $\mcE_{U_0}$ is given by
\begin{align}
    J_{\mcE_{U_0}} = J_{g_{U_0}} + \beta_{U_0} \otimes I,
\end{align}
where $N$ is given by $N=\tr \beta_{U_0}$.  Thus, the optimization problem (\ref{eq:optimization}) is rewritten as
\begin{align}
\begin{split}
    &\min \tr \beta_{U_0}\\
    \text{s.t. }& \tilde{J}_{g_{U_0}} + \beta_{U_0} \otimes I \geq 0.
\end{split}
\label{eq:optimization_choi}
\end{align}
\qed

\subsection{Proof of Cor.~\ref{cor:logdepth}: Restriction of the input unitary to be the logarithmic-depth unitaries}
\label{appendix:logdepth}
Due to the linearity of the constraint~\eqref{eq:optimization} with respect to $H\in \su(d)$, Eq.~\eqref{eq:optimization} remains unchanged even if we restrict $H$ to be a basis $\{P_j\}_j$ of $\su(d)$, i.e., restricting the input unitaries to be $\{e^{-iP_j\theta}\}_{\theta\in (-\epsilon, \epsilon), j}$ for a sufficiently small $\epsilon>0$.
As shown in the main text, such unitaries can be implemented with logarithmic-depth quantum circuits composed of a single-qubit $X$-rotation with angle $\theta$, two multi-target-CNOT gates, and two layers of single-qubit Clifford gates.
See Fig.~\ref{fig:pauli_rotation} for an example of $P_j = X\otimes Y\otimes Z$.

\begin{figure}
    \centering
    \includegraphics{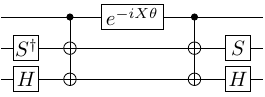}
    \caption{Quantum circuit for implementing $e^{-i(X\otimes Y\otimes Z)\theta}$ using a single-qubit $X$-rotation with angle $\theta$, two multi-target-CNOT gates, and two layers of single-qubit Clifford gates, where $H$ and $S$ are defined by $H\coloneqq {1\over \sqrt{2}}(\ketbra{0}{0}+\ketbra{0}{1}+\ketbra{1}{0}-\ketbra{1}{1})$, and $S\coloneqq \ketbra{0}{0}+i\ketbra{1}{1}$.}
    \label{fig:pauli_rotation}
\end{figure}

\section{Proof of Cor.~\ref{cor::inversion-transposistion-conjugation} in the main text}
\label{app:prooftable}
In this section, we show the lower bounds of the query complexity of unitary inversion, unitary transposition, and unitary complex conjugation shown in Cor.~\ref{cor::inversion-transposistion-conjugation} of the main text.
For an extra example, we also derive the lower bound of unitary iteration $f(U) = U^n$.

\subsection{Unitary inversion}
\noindent The primal SDP in Thm.~\ref{th::main_sdp}:
Since $g_{U_0}$ is given by $g_{U_0}(H)=-H$ for $H\in \su(d)$, $J_{g_{U_0}}$ is given by
\begin{align}
    J_{g_{U_0}}=
    -\dketbra{I}+\frac{1}{d}I\otimes I.
\end{align}
By setting $\beta_{U_0}= ((d^2-1)/d)I$, $\beta_{U_0}$ satisfies the SDP constraint since
\begin{align}
    &
    -\dketbra{I}+\frac{1}{d}I\otimes I
    +\frac{d^2-1}{d}I\otimes I
    \nonumber\\
    =&-\dketbra{I}+dI\otimes I\geq 0
\end{align}
{holds, and $\Tr \beta_{U_0}$ is given by}
 ${\rm tr}(\beta_{U_0})=d^2-1$.

This solution gives the minimum solution of the SDP, as shown below.
By taking the inner product of $\dketbra{I}$ with the SDP constraint given by
\begin{align}
    -\dketbra{I}+\frac{1}{d}I\otimes I + \beta_{U_0} \otimes I\geq 0,
\end{align}
we obtain
\begin{align}
    -d^2+1+\Tr\beta_{U_0} \geq 0,
\end{align}
i.e.,
\begin{align}
    \Tr\beta_{U_0} \geq d^2-1
\end{align}
holds. 

The lower bound $d^2-1$ can be made larger by $1$ by an extra discussion based on proof by contradiction. Suppose that inversion can be implemented by $d^2-1$ queries to $U$. Then there exists a $\beta_{U_0}$ with trace $d^2-1$ such that
\begin{align}
    J_{\mcE_{U_0}}=-\dketbra{I}+\frac{1}{d}I\otimes I +\beta_{U_0}\otimes I\geq 0.
\end{align}

On the other hand, $J_{\mcE_{U_0}}$ is originally defined as a Choi operator of $\mcE_{U_0}$ defined in Eq.~(\ref{eq::eU0_def_orig}), thus we have
\begin{align}
    &-\dketbra{I}+\frac{1}{d}I\otimes I +\beta_{U_0}\otimes I 
    \nonumber\\
    \geq &\sum_{j}\dketbra{M^{(1,{\rm left})}_{j,0}(U_0)^{\dagger}},
\end{align}
where $M^{(s,{\rm left})}_{j,0}(U_0)$ is defined as in Eq.~(\ref{eq::main_M_def}). Since $M^{(s, {\rm left})}_{j,0}(U_0)=U_0M^{(s, {\rm left})}_{j,0}(I)$ holds, we have
\begin{align}
    &-\dketbra{I}+\frac{1}{d}I\otimes I +\beta_{U_0}\otimes I 
    \nonumber\\
    \geq &
    (U_0^*\otimes I)\sum_{j}\dketbra{M^{(1,{\rm left})}_{j,0}(I)^{\dagger}}(U_0^{T}\otimes I).
\end{align}
By taking the average of $U_0$ over the Haar measure, we have
\begin{align}
    &-\dketbra{I}+\frac{1}{d}I\otimes I +\beta \otimes I 
    \nonumber\\
    \geq &
    \int {\rm d}U_0 (U_0^*\otimes I)\sum_{j}\dketbra{M^{(1,{\rm left})}_{j,0}(I)^{\dagger}}(U_0^{T}\otimes I)
    \nonumber\\
    =&\frac{1}{d}I\otimes \sum_j(M^{(1,{\rm left})}_{j,0}(I)^{\dagger} M^{(1,{\rm left})}_{j,0}(I))
    \nonumber\\
    =&\frac{1}{d}I\otimes I
    ,
\end{align}
where $\beta$ is the Haar average of $\beta_{U_0}$. The last equality follows from Eq.~(\ref{eq::leftleft}).
However, this inequality shows contradiction since
\begin{align}
    &\dbra{I}
    (-\dketbra{I}+\frac{1}{d}I\otimes I +\beta \otimes I )
    \dket{I}
    =0\\
    &<1 = \dbra{I}(\frac{1}{d}I\otimes I)\dket{I}
\end{align}
holds, and thus the assumption ${\rm tr}(\beta_{U_0}) = d^2-1$ is wrong.  Therefore, we obtain
${\rm tr}(\beta_{U_0})\neq d^2-1$, namely ${\rm tr}(\beta_{U_0})\geq d^2$.

\subsection{Unitary transposition}
\noindent The primal SDP in Thm.~\ref{th::main_sdp}:
Since {$g_{U_0}$ is given by} $g_{U_0}(H)=H^T$ for $H\in \su (d)$, $J_{g_{U_0}}$ is given by
\begin{align}
    J_{g_{U_0}}=\SWAP -\frac{1}{d}I\otimes I.
\end{align}
By setting $\beta_{U_0}=((d+1)/d)I$, {$\beta_{U_0}$ satisfies the SDP constraint since}
\begin{align}
    &\SWAP -\frac{1}{d}I\otimes I+
    \frac{d+1}{d}I\otimes I
    \nonumber\\
    =&2\Pi_{{\rm sym}} \geq 0
\end{align}
holds, 
where $\Pi_{\rm sym}$ is the projector onto the symmetric subspace, and $\Tr\beta_{U_0}$ is given by ${\rm tr}(\beta_{U_0})=d+1$. 

This solution gives the minimum solution of the SDP, as shown below.
We consider an orthogonal projector onto the antisymmetric subspace of $\CC^d\otimes \CC^d$ denoted by $\Pi_\mathrm{antisym}$.
By taking an inner product of $\Pi_\mathrm{antisym}$ with the SDP constraint given by
\begin{align}
    \SWAP -\frac{1}{d}I\otimes I + \beta_{U_0} \otimes I \geq 0,
\end{align}
we obtain
\begin{align}
    \Tr(\Pi_\mathrm{antisym}) (-d-1+\Tr\beta_{U_0})\geq 0,
\end{align}
i.e.,
\begin{align}
    \Tr \beta_{U_0}\geq d+1
\end{align}
holds.

The lower bound $d+1$ can be made larger by 1 ($d=2$) and 2 ($d\geq 3$) by an extra discussion. According to Eq.~(\ref{eq::eU0_def_orig}), $J_{\mcE_{U_0}}$ is expressed as
\begin{align}
    &J_{\mcE_{U_0}}=\SWAP -\frac{1}{d}I\otimes I + \beta_{U_0}\otimes I
    \nonumber\\
    =&\sum_{s=1}^{N}
    \sum_j\dketbra{M^{(s, {\rm left})}_{j,0}(U_0)^{\dagger}}.
\end{align}
Defining $Q^{(s)}_{j,k}$ as $V_s\eqqcolon\sum_{j,k}Q^{(s)}_{j,k}\otimes \ketbra{j}{k}$, we obtain
\begin{align}\label{eq::bef_twirl_trans}
    &\SWAP -\frac{1}{d}I\otimes I + \beta_{U_0}\otimes I
    \nonumber\\
    \geq&
    \sum_{s=1}^2\sum_j\dketbra{M^{(s, {\rm left})}_{j,0}(U_0)^{\dagger}}
    \nonumber\\
    =&
    \sum_j\dketbra{(Q^{(1)}_{j,0})^{\dagger}U_0^{\dagger}}
    \nonumber\\
    +&\sum_{j,k,\ell}
    \dket{(Q^{(1)}_{k,0})^{\dagger}U_0^{\dagger}(Q^{(2)}_{j,k})^{\dagger}U_0^{\dagger}}\dbra{(Q^{(1)}_{\ell,0})^{\dagger}U_0^{\dagger}(Q^{(2)}_{j,\ell})^{\dagger}U_0^{\dagger}}
    \nonumber\\
    =& (U_0^*\otimes I) \sum_j\dketbra{(Q^{(1)}_{j,0})^{\dagger}}
    (U_0^T\otimes I)
    \nonumber\\
    +&\sum_{j,k,\ell}
    (U_0^*\otimes (Q^{(1)}_{k,0})^{\dagger}U_0^{\dagger})
    \dket{(Q^{(2)}_{j,k})^{\dagger}}\dbra{(Q^{(2)}_{j,\ell})^{\dagger}}
    (U_0^T\otimes U_0Q^{(1)}_{\ell,0}) .
\end{align}
Taking the Haar average with $U_0$, the left-hand side of Eq.~(\ref{eq::bef_twirl_trans}) is rewritten as
\begin{align}\label{eq::q_takusan}
    &\frac{1}{d}I\otimes \sum_j (Q^{(1)}_{j,0})^{\dagger}Q^{(1)}_{j,0}
    \nonumber\\
    +&\frac{1}{d^2-1}\sum_{j,k,\ell}
    \left[
    {\rm tr}((Q^{(2)}_{j,k})^{\dagger}Q^{(2)}_{j,\ell})I\otimes (Q^{(1)}_{k,0})^{\dagger}Q^{(1)}_{\ell,0}
    \right.
    \nonumber\\
    &\quad\quad\quad\quad\quad -\frac{1}{d}I\otimes 
    (Q^{(1)}_{k,0})^{\dagger}Q^{(2)}_{j,\ell}(Q^{(2)}_{j,k})^{\dagger}Q^{(1)}_{\ell,0}
    \nonumber\\
    &\quad\quad\quad\quad\quad +
    \dket{(Q^{(1)}_{k,0})^{\dagger}Q^{(2)}_{j,\ell}}\dbra{(Q^{(1)}_{\ell, 0})^{\dagger}Q^{(2)}_{j,k}}
    \nonumber\\
    &\left.
    \quad\quad\quad\quad\quad -
    \frac{1}{d}(Q^{(2)}_{j,\ell})^T(Q^{(2)}_{j,k})^*\otimes (Q^{(1)}_{k,0})^{\dagger}Q^{(1)}_{\ell,0}
    \right]
    \nonumber\\
    =&\frac{1}{d}I\otimes I
    +\frac{1}{d^2-1}
    \left[
    (d-\frac{1}{d})I\otimes I
    \right.\nonumber\\
    &
    -\frac{1}{d}\sum_{j,k,\ell}
    I\otimes (Q^{(1)}_{k,0})^{\dagger}Q^{(2)}_{j,\ell}(Q^{(2)}_{j,k})^{\dagger}Q^{(1)}_{\ell,0}
    \nonumber\\
    &\left.
    +\sum_{j,k,\ell}
    \dket{(Q^{(1)}_{k,0})^{\dagger}Q^{(2)}_{j,\ell}}\dbra{(Q^{(1)}_{\ell, 0})^{\dagger}Q^{(2)}_{j,k}}
    \right].
\end{align}
Here, the following formulae
\begin{align}
    &\int{\rm d}U\ UMU^{\dagger}=\frac{{\rm tr}M}{d}I,
    \nonumber\\
    &\int{\rm d}U\ (U\otimes U^T) M_{12} (U^{\dagger}\otimes U^*)
    \nonumber\\
    &~~~~~~=
    \frac{1}{d^2-1}\left[
    ({\rm tr}M_{12})I\otimes I
    -\frac{1}{d}I\otimes {\rm tr}_1(\tilde{M}_{12}) \right. \nonumber \\
    &~~~~~~~~~~~~~~~~~~~~~+ \left. \tilde{M}_{12}-\frac{1}{d}({\rm tr}_2\tilde{M}_{12})\otimes I
    \right],
    \nonumber\\
    &~~~~~~~~~(\tilde{M}_{12}\coloneqq({\rm SWAP})M_{12}^T({\rm SWAP}))
    \nonumber\\
    &\sum_j(Q^{(2)}_{j,k})^{\dagger}Q^{(2)}_{j,\ell}=\delta_{k,\ell}I,
    \nonumber\\
    &\sum_j(Q^{(1)}_{j,0})^{\dagger}Q^{(1)}_{j,0}=I
\end{align}
are used. By taking the inner product with $\Pi_{\rm antisym}$, we have
\begin{align}\label{eq::d+3_bef}
    &-\frac{d(d-1)}{2}-\frac{d-1}{2}+\frac{d-1}{2}{\rm tr}\beta_{U_0}
    \nonumber\\
    \geq &
    \frac{d-1}{2} +\frac{1}{d^2-1}\left[
    \frac{(d-1)(d^2-1)}{2} 
    \right.
    \nonumber\\
    -&\frac{1}{2}
    \left(-\frac{1}{d}\sum_{j,k,\ell} {\rm tr}((Q^{(1)}_{k,0})^{\dagger}Q^{(2)}_{j,\ell}(Q^{(2)}_{j,k})^{\dagger}Q^{(1)}_{\ell,0})
    \right.\nonumber\\
    &\left.\left.
    +\sum_{j,k,\ell} {\rm tr}((Q^{(1)}_{k,0})^{\dagger}Q^{(2)}_{j,\ell}(Q^{(1)}_{\ell,0})^T(Q^{(2)}_{j,k})^{*})
    \right)
    \right].
\end{align}
Since the second term $(1/(d^2-1))[\cdots]$ is obtained as a Hilbert Schmidt inner product of two positive operators and thus is nonnegative, we have
\begin{align}
    -\frac{d(d-1)}{2}-\frac{d-1}{2}+\frac{d-1}{2}{\rm tr}\beta_{U_0} \geq {d-1\over 2},
\end{align}
i.e.,
\begin{align}
    {\rm tr}\beta_{U_0}\geq d+2.
\end{align}

Also, the second term $(1/(d^2-1))[\cdots]$ of Eq.~(\ref{eq::d+3_bef}) is lower-bounded by $(d-1)/2-d/(2(d-1))$ which is larger than 0 for $d\geq 3$. This can be proved by noticing 
\begin{align}
    &
    -\frac{1}{d}\sum_{j,k,\ell} {\rm tr}((Q^{(1)}_{k,0})^{\dagger}Q^{(2)}_{j,\ell}(Q^{(2)}_{j,k})^{\dagger}Q^{(1)}_{\ell,0})
    \nonumber\\
    &
    +\sum_{j,k,\ell} {\rm tr}((Q^{(1)}_{k,0})^{\dagger}Q^{(2)}_{j,\ell}(Q^{(1)}_{\ell,0})^T(Q^{(2)}_{j,k})^{*})
    \nonumber\\
    =&-\frac{1}{d}{\rm tr}(A^{\dagger}B)+{\rm tr}(A^{\dagger}C)
    \nonumber\\
    \leq & \frac{1}{d}\|A\|_2\|B\|_2+\|A\|_2\|C\|_2 = d(d+1)
\end{align}
for 
\begin{align}
    A&\coloneqq\sum_{j,k,\ell}(Q^{(2)}_{j,\ell})^{\dagger}Q^{(1)}_{k,0}\otimes \ket{j,k,\ell},
    \nonumber\\
    B&\coloneqq\sum_{j,k,\ell}(Q^{(2)}_{j,k})^{\dagger}Q^{(1)}_{\ell,0}\otimes \ket{j,k,\ell},
    \nonumber\\
    C&\coloneqq \sum_{j,k,\ell}
    (Q^{(1)}_{\ell,0})^T(Q^{(2)}_{j,k})^*\otimes \ket{j,k,\ell},
\end{align}
and that the 2-norm of $A,\ B,\ C$ are $d$. Therefore, for $d\geq 3$, we obtain
\begin{align}
    {\rm tr}\beta_{U_0}\geq d+3-\frac{d}{2(d-1)}>d+2.
\end{align}

\subsection{Unitary complex conjugation}
\noindent The primal SDP in Thm.~\ref{th::main_sdp}:
Since {$g_{U_0}$ is given by} $g_{U_0}(H)=-U_0^TH^*U_0^*$ for $H\in \su (d)$, $J_{g_{U_0}}$ is given by
\begin{align}
    J_{g_{U_0}}=-(I\otimes U_0^T)\left(\SWAP -\frac{1}{d}I\otimes I\right)(I\otimes U_0^*).
\end{align}
By setting $\beta_{U_0}=((d-1)/d)I$, {$\beta_{U_0}$ satisfies the SDP constraint since}
\begin{align}
    &-(I\otimes U_0^T)\left(\SWAP -\frac{1}{d}I\otimes I\right)(I\otimes U_0^*)+\frac{d-1}{d}I\otimes I
    \nonumber\\
    =&
    2(I\otimes U_0^T)
    \Pi_{\rm antisym}
    (I\otimes U_0^*) \geq 0
\end{align}
{holds, }
and {$\Tr\beta_{U_0}$ is given by} ${\rm tr}(\beta_{U_0})=d-1$.

This solution gives the minimum solution of the SDP, as shown below.
We consider an orthogonal projector onto the symmetric subspace of $\CC^d\otimes \CC^d$ denoted by $\Pi_\mathrm{sym}$.
By taking an inner product of $(I\otimes U^T_0)\Pi_\mathrm{sym}(I\otimes U^*_0)$ with the SDP constraint given by
\begin{align}
    -(I\otimes U_0^T)\left(\SWAP -\frac{1}{d}I\otimes I\right)(I\otimes U_0^*)+\beta_{U_0}\otimes I\geq 0,
\end{align}
we obtain
\begin{align}
    \Tr(\Pi_\mathrm{sym}) (-(d-1) + \Tr \beta_{U_0})\geq 0,
\end{align}
i.e.,
\begin{align}
    \Tr\beta_{U_0}\geq d-1
\end{align}
holds.  This bound is achievable by the construction of an algorithm given by \cite{miyazaki2019complex}. Therefore, it is tight.  This proof is an alternative proof of the tight optimal lower bound $d-1$ originally shown in \cite{quintino2019reversing}.   

\subsection{Unitary iteration}
Unitary iteration is a task to transform $U\in \SU(d)$ to $f(U)=U^n$.\\
\noindent The primal SDP in Thm.~\ref{th::main_sdp}: 
Since {$g_{U_0}$ is given by} $g_{U_0}(H)=\sum_{k=1}^n U_0^{-k}HU_0^k$ for $H\in \SU (d)$, $J_{g_{U_0}}$ is given by
\begin{align}
    J_{g_{U_0}}=\sum_{k=1}^n \dketbra{U_0^{-k}}-\frac{n}{d}I\otimes I.
\end{align}
By setting $\beta_{U_0}=(n/d)I$, {$\beta_{U_0}$ satisfies the SDP constraint since} 
\begin{align}
    \sum_{k=1}^n \dketbra{U_0^{-k}}\geq 0
\end{align}
holds, and $\Tr\beta_{U_0}$ is given by ${\rm tr}\beta_{U_0}=n$. 

This solution gives the minimum, as shown below.
The SDP constraint is given by
\begin{align}
    \sum_{k=1}^n \dketbra{U_0^{-k}}-\frac{n}{d}I\otimes I + \beta_{U_0} \otimes I\geq 0.\label{eq:iteration_sdp_constraint}
\end{align}
We define orthogonal projectors $\Pi_j$ on $\CC^d$ using the eigendecomposition of $U_0$ given by
\begin{align}
    U_0 = \sum_{j=1}^{d} e^{i\phi_j} \Pi_j,
\end{align}
where $e^{i\phi_j}$ for $\phi_j\in\RR, j\in\{1, \cdots, d\}$ is the $j$-th eigenvalue of $U_0$, and $\Pi_j$ is the orthonormal projector onto the corresponding eigenvector.
The set of the dual vectors $\{\dket{\Pi_j}\}_{j=1}^{d}$ forms an orthonormal basis of $\mathrm{span}\{\dket{\Pi_j}\}$ since
\begin{align}
    \dbraket{\Pi_j}{\Pi_k} = \Tr(\Pi_j^\dagger \Pi_k) = \delta_{jk}
\end{align}
holds, where $\delta_{jk}$ is Kronecker's delta defined by $\delta_{jj} =1$ and $\delta_{jk} = 0$ for $j\neq k$.
The orthogonal projector onto the complement of $\mathrm{span}\{\dket{\Pi_j}\}$ given by
\begin{align}
    \Pi^\perp\coloneqq I\otimes I-\sum_{j=1}^{d}\dketbra{\Pi_j}
\end{align}
satisfies
\begin{align}
    \Tr(\Pi^\perp \dketbra{U_0^{-k}}) = 0,\\
    \Tr_2 \Pi^\perp = (d-1) I.
\end{align}
Thus, taking the inner product of $\Pi^\perp$ with Eq.~\eqref{eq:iteration_sdp_constraint}, we obtain
\begin{align}
    (d-1)(-n+\Tr\beta_{U_0})\geq 0,
\end{align}
i.e.,
\begin{align}
    \Tr \beta_{U_0}\geq n
\end{align}
holds.

\section{Modification of Thm.~\ref{th::main_sdp} in the main text to a subgroup of $\SU(d)$}

We derive a lower bound on the query complexity of a differentiable function $f:S\to S$ for a Lie subgroup $S$ of $\SU(d)$.
We define a linear map $g_{U_0}:\mathfrak{s}\to \mathfrak{s}$ by
\begin{align}\label{eq::diff_sub}
    g_{U_0} (H)\coloneqq
    -i\left.\frac{{\rm d}}{{\rm d}\epsilon}\right|_{\epsilon =0}\left[f(U_0)^{-1}f(e^{i\epsilon H}U_0)\right],
\end{align}
where $\mathfrak{s}$ is a Lie algebra of $S$.
We define orthonormal bases $\{G_j\}_j$ and $\{B_k\}_k$ of $\mathfrak{s}$ and $\su(d)\setminus \mathfrak{s}$, respectively.
Then, we show the following Theorem.

\label{app::subgroup}
\begin{Theorem}\label{th::sdp_subgroup}
For any differentiable function $f: S\to S$, the query complexity of $f$ is larger than or equal to the solution of the following SDP:
\begin{align}\label{eq::sdp_subgroup}
\begin{split}
    &\min_{\{B'_k\}_k} \tr \beta_{U_0}\\
    \text{s.t. }& \tilde{J}_{g_{U_0}} + \beta_{U_0} \otimes I \geq 0,\\
    &\tilde{J}_{g_{U_0}} \coloneqq \sum_{j} G_j^* \otimes g_{U_0}(G_j)+\sum_{k}B_k^*\otimes B'_k,
\end{split}
\end{align}
where $U_0$ is an arbitrary unitary operator in $S$, $B'_k$ is an arbitrary traceless $d\times d$ operator.
\end{Theorem}
A trivial upper bound of the solution of the SDP~\eqref{eq::sdp_subgroup} is found by setting $B'_k=0$ and $\beta_{U_0}\coloneqq|\lambda|I$ where $\lambda <0$ is the minimum eigenvalue of $\sum_j G_j^*\otimes g_{U_0}(G_j)$.
Since $\lambda^2\leq \|\sum_{j}G_j^*\otimes g_{U_0}(G_j)\|_2^2=\sum_j\|g_{U_0}(G_j)\|_2^2$ holds, an upper bound of $N$ is given by ${\rm tr}\beta_{U_0}=d\sqrt{\sum_j\|g_{U_0}(G_j)\|_2^2}$, which potentially implies that the deterministic and exact implementation of $f$ can be achieved by a smaller number of queries if $U$ is limited to a small subgroup.

The dual of SDP in Eq.~(\ref{eq::sdp_subgroup}) is

\begin{align}
\begin{split}
    &\max -\Tr[(\sum_j G_j^* \otimes g_{U_0}(G_j))\Gamma]\\
    \mathrm{s.t.}\;& \Gamma\geq 0,\\
    &\Tr_1[(B_k^*\otimes I) \Gamma] = 0 \quad \forall k,\\
    &\Tr_2\Gamma = I,
\end{split}
\end{align}
as shown in Appendix \ref{app::duals}.

\textbf{Proof of Thm.~\ref{th::sdp_subgroup}:}
Let us choose some $U_0\in S$ and define $\tilde{V}_{N+1}(U_0)$ in the same way as in Fig.~\ref{fig::lowerbound_simplified}. Then, the same proof for Lem.~\ref{le::Vlefrig} holds for this case, thus we have
\begin{align}\label{eq::righttoleftprob}
    M^{(s,{\rm right})}_{0,j}(U_0)^{\dagger}
    =
    M^{(s,{\rm left})}_{j,0}(U_0),
\end{align}
where the notation follows that in Appendix~\ref{app::theo1}. 
Additionally, by substituting $U\gets e^{i\epsilon H}U_0$ for some $H\in {\rm span}(\{G_j\}_j)$ and differentiating by $\epsilon$ at $\epsilon=0$, we show that
\begin{align}\label{eq::comb_vs_f_prob}
\begin{cases}
    \mcE_{U_0}(I)=NI,\\
    \mcE_{U_0}(H)=g_{U_0}(H)+\alpha_{U_0}(H)I&(H\in \{G_j\}_j)
\end{cases}
\end{align}
hold for 
\begin{align}
    \mcE_{U_0}(H)=
    \sum_{s=1}^{N}
    \sum_j
    (M^{(s,{\rm left})}_{j,0}(U_0))^{\dagger} H (M^{(s,{\rm left})}_{j,0}(U_0))
\end{align}
in the same way as the proof of Lem.~\ref{le::CP_comb}. On the other hand, the action of $\mcE_{U_0}$ on $H\notin {\rm span}(\{G_j\}_j)$ is not determined, thus a lower bound of the number of queries to $U$ is given by
\begin{align}
\label{eq:sdp_sub_proof}
\begin{split}
    &\min_{\{B'_k\}_k} N\\
    \text{s.t. }& \mcE_{U_0} \text{ is CP}, \alpha_{U_0}: \su(d) \to \su(d) \text{ is linear},\\
    & \mcE_{U_0}(I)=NI,\\
    & \mcE_{U_0}(G_j) = g_{U_0}(G_j) + \alpha_{U_0}(G_j)I,\\
    & \mcE_{U_0}(B_k) = B'_k +\alpha_{U_0}(B_k)I,
\end{split}
\end{align}
where $B'_k$ are arbitrary traceless operators.
By defining $\beta_{U_0}$ satisfying
\begin{align}
\begin{split}
    &\tr(\beta_{U_0}^T I)=N,\\
    &\tr(\beta_{U_0}^T H)=\alpha_{U_0}(H)\quad (H\in \su (d)),
\end{split}
\end{align}
the Choi operator of $\mcE_{U_0}$ is given by
\begin{align}
    J_{\mcE_{U_0}}=\tilde{J}_{g_{U_0}}+\beta_{U_0}\otimes I.
\end{align}
Therefore, the SDP~\eqref{eq:sdp_sub_proof} reduces to the SDP~\eqref{eq::sdp_subgroup}, which shows Thm.~\ref{th::sdp_subgroup}. \qed

In what follows, we consider three examples of the subgroups as follows:
\begin{itemize}
    \item $\SU(d)^{\otimes n}\coloneqq \{U_1\otimes \cdots \otimes U_n \mid U_1, \cdots, U_n\in\SU(d)\}\subset \SU(d^n)$
    \item Diagonal unitary inversion
    \item $\mathrm{SO}(d)\subset \SU(d)$ for unitary inversion
\end{itemize}

\subsection{$\SU(d)^{\otimes n} \subset \SU(d^n)$}
One simple example of a subgroup of a unitary group is the tensor product of unitary operations, e.g., $\SU(d)^{\otimes n}\subset \SU(d^n)$. The solution of the primal and dual SDP for this case satisfies the following property.
\begin{lem}
    The minimum value of the SDP in Thm.~\ref{th::main_sdp} in the main text at a unitary operation $U_0\in \SU (d)$ 
    for a function $f$ on $\SU(d)$ matches the minimum value of the SDP in Eq.~(\ref{eq::sdp_subgroup}) at a unitary operation $U_0^{\otimes n}$ for a function $F: \SU(d)^{\otimes n} \to \SU(d)^{\otimes n}$ defined as 
    \begin{align}
        F\left(\bigotimes_{j=1}^n U_j\right)\coloneqq
        \bigotimes_{j=1}^n f(U_j).
    \end{align}
\end{lem}
When the transformation $f$ can be implemented by $N$ queries to $U$, then $F$ can also be implemented in the same number of queries $N$ (in fact, by running the transformation circuit in parallel, $\bigotimes_j U_j$ can be transformed into $\bigotimes_j f(U_j)$), which does not scale on the total dimension $d^n$ for a fixed $d$. This lemma shows that the SDP in Eq.~(\ref{eq::sdp_subgroup}) correctly captures the independence of the query numbers on $n$.  

\textbf{Proof:}~
Let us define the Hilbert space of input $\bigotimes_jU_j$ and output $\bigotimes_jf(U_j)$ unitary operators as $\bigotimes_j \mcH_j$ and $\bigotimes_j \mcH_j'$, respectively.
The Lie algebra of $\SU (d)^{\otimes n}$ is spanned by an orthonormal basis $(1/\sqrt{d^{n-1}})(G_j)_{\mcH_l}\otimes \bigotimes_{m\neq l}(I)_{\mcH_m}$ using an orthonormal basis $\{G_j\}$ of $\SU(d)$ and the corresponding value of differentiation of Eq.~(\ref{eq::diff_sub}) is given by $(1/\sqrt{d^{n-1}})(g(G_j))_{\mcH'_l}\otimes \bigotimes_{m\neq l}(I)_{\mcH'_m}$. 
Thus, the SDP \eqref{eq::sdp_subgroup} for the function $F$ is given by
\begin{align}\label{eq:sdp_tensor}
\begin{split}
    &\min \Tr \hat{\beta}_{U_0}\\
    \mathrm{s.t.}\;& \hat{\beta}_{U_0}\in \bigotimes_{j}\mcL(\mcH_j), B'_k\in \bigotimes_{j}\mcL(\mcH'_j),\\
    &\sum_{l=1}^{n} (J_{g_{U_0}})_{\mcH_l \mcH'_l} \otimes \bigotimes_{m\neq l} {(I\otimes I)_{\mcH_m\mcH'_m} \over d}\\
    &+ \sum_{k\in K} B_k^* \otimes B'_k + \hat{\beta}_{U_0} \otimes I_{\mcH'}\geq 0,\\
    &\Tr B'_k = 0 \quad \forall k\in K,
\end{split}
\end{align}
where $J_{g_{U_0}}$ is given in Thm.~\ref{th::main_sdp} and $B_k$ is given by
\begin{align}
    B_k \coloneqq \bigotimes_{l} (G_{k_l})_{\mcH_l},
\end{align}
where $G_0\coloneqq I/\sqrt{d}$ and the summand over $k=(k_1, \cdots, k_n)$ is taken over the set $K\coloneqq \{k|\#(l|k_l\neq 0)\geq 2\}$, where $\#(l|k_l\neq 0)$ represents the number of $l$'s such that $k_l \neq 0$.

Suppose $\beta_{U_0}$ is a solution of the SDP~\eqref{eq:sdp} in the main text, i.e.,
\begin{align}
    J_{g_{U_0}} + \beta_{U_0} \otimes I\geq 0
\end{align}
holds.
Then, defining $\hat{\beta}_{U_0}$ and $B'_k$ by
\begin{align}
    \hat{\beta}_{U_0} &\coloneqq \sum_{l=1}^{n} (\beta_{U_0})_{\mcH_l} \otimes \bigotimes_{m\neq l} {I_{\mcH_m} \over d},\\
    B'_k&\coloneqq 0,
\end{align}
$\hat{\beta}_{U_0}$ and $B'_k$ give a solution of the SDP \eqref{eq:sdp_tensor} and $\Tr\hat{\beta}_{U_0} = \Tr \beta_{U_0}$ holds.
Therefore, the solution of the SDP~\eqref{eq:sdp_tensor} is upper bounded by the solution of the SDP~\eqref{eq:sdp} in the main text.

Conversely, suppose $\hat{\beta}_{U_0}$ and $B'_k$ give a solution of the SDP \eqref{eq:sdp_tensor}, then
\begin{align}
    &\sum_{l=1}^{n} (J_{g_{U_0}})_{\mcH_l \mcH'_l} \otimes \bigotimes_{m\neq l} {(I\otimes I)_{\mcH_m\mcH'_m} \over d}\nonumber\\
    &+ \sum_{k\in K} B_k^* \otimes B'_k + \hat{\beta}_{U_0} \otimes I_{\mcH'}\geq 0
\end{align}
holds.
Taking the partial trace on $\mcH_{\neq 1}\otimes \mcH'_{\neq 1} \coloneqq \bigotimes_{m\neq 1} \mcH_m \otimes \mcH'_m$, we obtain
\begin{align}
    d^{n-1} (J_{g_{U_0}} + \Tr_{\mcH_{\neq 1}\mcH'_{\neq 1}} \hat{\beta}_{U_0} \otimes I_{\mcH_1})\geq 0,
\end{align}
where we use the identities
\begin{align}
    \Tr J_{g_{U_0}} &= 0,\\
    \Tr_{\mcH_{\neq 1}} B_k^* &= 0.
\end{align}
Thus, defining $\beta_{U_0}$ by $\beta_{U_0}\coloneqq \Tr_{\mcH_{\neq 1}\mcH'_{\neq 1}} \hat{\beta}_{U_0}$, $\beta_{U_0}$ is a solution of the SDP~\eqref{eq:sdp} in the main text and $\Tr\beta_{U_0} = \Tr\hat{\beta}_{U_0}$ holds.
Therefore, the solution of the SDP~\eqref{eq:sdp} in the main text is upper bounded by the solution of the SDP \eqref{eq:sdp_tensor}.
In conclusion, the SDP \eqref{eq:sdp_tensor} gives the same minimum value as the SDP~\eqref{eq:sdp} in the main text.
\qed

{\subsection{Diagonal unitary inversion}}

The Lie algebra of the diagonal unitary is given by
\begin{align}
    \mathfrak{s} = \mathrm{span}\{Z^k \mid k=1, \cdots, d-1\},
\end{align}
where $Z$ is the clock operator defined by $Z\coloneqq \sum_j \omega^j \ketbra{j}{j}$ for $\omega\coloneqq e^{2\pi i/n}$ and the computational basis $\{\ket{j}\}$ of $\CC^d$.
The complement $\mathfrak{su}(d)\setminus \mathfrak{s}$ is given by
\begin{align}
    \mathfrak{su}(d)\setminus \mathfrak{s} = \mathrm{span}\{X^j Z^k \mid j\in\{1, \cdots, d-1\}, \nonumber \\
    k\in\{0, \cdots, d-1\}\},
\end{align}
where $X$ is the shift operator defined by $X\coloneqq \sum_j\ketbra{j\oplus 1}{j}$.
Thus, the SDP \eqref{eq::sdp_subgroup} for the diagonal unitary inversion is given by
\begin{align}
\begin{split}
    &\min \Tr\beta_{U_0}\\
    \mathrm{s.t.}\;&\beta_{U_0}\in\mcL(\CC^d), B'_{jk}\in\mcL(\CC^d),\\
    &- \sum_{k=1}^{d-1}{Z^{-k}\otimes Z^k \over d} + \sum_{j=1}^{d-1}\sum_{k=0}^{d-1} {X^{j}Z^{-k} \over \sqrt{d}} \otimes B'_{jk} +\beta_{U_0} \otimes I \nonumber \\
    &~~~~~~~~~~~~~~~~~~~~~~\geq 0,\\
    &\Tr B'_{jk} = 0 \quad \forall j, k\in\{1, \cdots, d-1\}.
\end{split}
\end{align}
Note that $\sum_{k=1}^{d-1}Z^{-k}\otimes Z^k$ is given by
\begin{align}
    &\sum_{k=1}^{d-1}Z^{-k}\otimes Z^k\nonumber\\
    &= \sum_{k=0}^{d-1} Z^{-k}\otimes Z^k - I\otimes I\\
    &= \sum_{k=0}^{d-1} \sum_{j_1, j_2=1}^{d} \omega^{-k(j_1-j_2)}\ketbra{j_1}\otimes\ketbra{j_2}-I\otimes I\\
    &= d {\sum_{j=1}^{d} \ketbra{jj}}-I\otimes I.
\end{align}
By setting $\beta_{U_0} = {d-1\over d} I$ and $B'_{jk} = 0$, $\beta_{U_0}$ and $B'_{jk}$ satisfy the SDP constraints, and $\Tr\beta_{U_0} = d-1$.

This solution gives the minimum, as shown below.
Taking the diagonal components of the second constraint, we obtain
\begin{align}\label{eq:sdp_constraint_diagonal}
    {1\over d} I\otimes I - {\sum_{j=1}^{d} \ketbra{jj}} +[\beta_{U_0}]_\mathrm{diag} \otimes I\geq 0,
\end{align}
where $[\beta_{U_0}]_\mathrm{diag}$ is the matrix obtained by setting the off-diagonal components of $\beta_{U_0}$ to be zero.
By taking the inner product of ${\sum_{j=1}^{d} \ketbra{jj}}$ with Eq.~\eqref{eq:sdp_constraint_diagonal}, we obtain
\begin{align}
    1-d+\Tr\beta_{U_0}\geq 0,
\end{align}
i.e.,
\begin{align}
    \Tr\beta_{U_0} \geq d-1.
\end{align}

{\subsection{$\mathrm{SO}(d)\subset \SU(d)$ for unitary inversion}}
As another example, we obtain the lower bound $d-1$ for unitary inversion restricted to the subgroup ${\rm SO}(d)$ of $\SU (d)$. If this bound is tight, then the optimal scheme of inverting an orthogonal operator in ${\rm SO} (d)$ is not by unitary transposition, which requires at least $d+3$ queries ($d\geq 3$) to $U$. Indeed, for $d=2$, the optimal scheme is not by transposing $U$ using $4$ queries to $U$, but by sandwiching $U$ by $X$, thus the lower bound $d-1$ is tight (note that an arbitrary $U\in {\rm SO}(2)$ is expressed as $\cos (\theta )I+i\sin (\theta)Y$ ($\theta\in [0,2\pi)$)).

\textbf{Proof:}~
The Lie algebra $\mathfrak{so}(d)$ of $\mathrm{SO}(d)$ is given by
\begin{align}
    \mathfrak{so}(d) = \mathrm{span}_\RR\{i\ketbra{j_1}{j_2}-i\ketbra{j_2}{j_1} \mid 1\leq j_1<j_2\leq d\}
\end{align}
and its complement $\mathfrak{su}(d)\setminus \mathfrak{so}(d)$ is given by
\begin{align}
    \mathfrak{su}(d)\setminus\mathfrak{so}(d)
    =& \mathrm{span}\{\ketbra{j_1}{j_2}+\ketbra{j_2}{j_1} \mid 1\leq j_1<j_2\leq d\}\nonumber\\
    &\oplus \mathrm{span}\{\ketbra{j}-I/d \mid 1\leq j\leq d\}.
\end{align}
Thus, the SDP \eqref{eq::sdp_subgroup} for the $\mathrm{SO}(d)$ unitary inversion is given by
\begin{align}
\begin{split}
    &\min \Tr \beta_{U_0}\\
    \mathrm{s.t.}\;&\beta_{U_0}\in\mcL(\CC^d), B'_{j_1 j_2}, B'_j\in\mcL(\CC^d),\\
    &\sum_{j_1<j_2}[{-(\ketbra{j_1}{j_2}-\ketbra{j_2}{j_1})^{\otimes 2}\over 2}+(\ketbra{j_1}{j_2}+\ketbra{j_2}{j_1}) \otimes B_{j_1 j_2}'] \\
    &+ \sum_j (\ketbra{j}-{I\over d}) \otimes B'_j + \beta_{U_0} \otimes I\geq 0,\\
    &\Tr B_{j_1j_2}' = \Tr B_{j}' = 0.
\end{split}
\end{align}
Assuming $B'_{j_1 j_2} = a (\ketbra{j_1}{j_2}+\ketbra{j_2}{j_1})$, $B'_j = 2a(\ketbra{j}-I/d)$ and $\beta_{U_0} = b I$ for $a = {d-2 \over 2(d+2)}$ and $b = {d-1\over d}$, the SDP constraint is satisfied since
\begin{align}
    &\sum_{j_1<j_2}[{-(\ketbra{j_1}{j_2}-\ketbra{j_2}{j_1})^{\otimes 2}\over 2}+(\ketbra{j_1}{j_2}+\ketbra{j_2}{j_1}) \otimes B_{j_1 j_2}'] \nonumber\\
    &+ \sum_j (\ketbra{j}-{I\over d}) \otimes B'_j + \beta_{U_0} \otimes I\\
    &= {d\over d+2}\mathrm{SWAP} -{2\over d+2} \dketbra{I} + {d\over d+2}I\otimes I\\
    &= {2d\over d+2}(\Pi_\mathrm{sym}-{1\over d}\dketbra{I})\\
    &\geq 0
\end{align}
holds.
In this case, $\Tr\beta_{U_0}$ is given by $\Tr\beta_{U_0} = d-1$.

This solution gives the minimum, as shown below.
By taking the inner product of $2\Pi_\mathrm{antisym} + \dketbra{I}$ with the second constraint and using the relations
\begin{align}
    \Tr(2\Pi_{\mathrm{antisym}}A\otimes B) &= \Tr(A)\Tr(B)-\Tr(AB),\\
    \Tr(\dketbra{I}A\otimes B) &= \Tr(A^T B),
\end{align}
we obtain
\begin{align}
    -d(d-1)+d\Tr \beta_{U_0} \geq 0,
\end{align}
i.e.,
\begin{align}
    \Tr\beta_{U_0}\geq d-1
\end{align}
holds.
\qed

\section{Modification of Thm.~\ref{th::main_sdp} to probabilistic case}
\label{app::prob}

\begin{Theorem}\label{le::prob_exact}
    Given any differentiable function $f:\SU(d)\to\SU(d)$ and a differentiable function $p: \SU(d) \to \RR_{\geq 0}$, the query complexity to implement $f$ with a probability greater than or equal to $p(U)>0$ in a neighborhood of a unitary operator $U_0$ is lower bounded by the solution of the following semidefinite programming (SDP):
\begin{align}\label{eq::sdp_for_prob}
\begin{split}
    &\min \tr \beta_{U_0}\\
    \text{s.t. }& 
    J_{\mcA}-J_{\mcB}=
    J_{g_{U_0}} + \beta_{U_0} \otimes I,\\
    &{\rm tr}J_{\mcB}=\frac{1-\sqrt{p(U_0)}}{1+\sqrt{p(U_0)}}
    {\rm tr}J_{\mcA},\\
    &J_{\mcA}, J_{\mcB}\geq 0,
\end{split}
\end{align}
where $J_{g_{U_0}}$ is the same operator as defined in Eq.~(2) in the main text.
\end{Theorem}
The dual SDP of the SDP above is written as
\begin{align}
\begin{split}
    \max &- \Tr(J_{g_{U_0}}M)\\
    \text{s.t. }& M\in \mcL(\mathbb{C}^d\otimes\mathbb{C}^d),\\
    & a\in \mathbb{R},\\
    & M-a\frac{1-\sqrt{p(U_0)}}{1+\sqrt{p(U_0)}}I\geq 0,\\
    & aI-M\geq 0,\\
    &\Tr_{2}M = I,
\end{split}
\end{align}
as shown in Appendix \ref{app::duals}.
For generality, we presented a theorem applicable to the situation where the success probability can depend on the unitary $U$. Similarly to the case of deterministic and exact transformation, the SDP in Eq.~(\ref{eq::sdp_for_prob}) for a unitary $U_0$ gives a necessary condition for the circuit to implement $f$ at the neighborhood of $U_0$ up to the first order of differentiation. 

An important property of SDP in Eq.~(\ref{eq::sdp_for_prob}) is that $N$ is a non-decreasing function of $p(U_0)\in [0,1]$ since the space of $J_{\mcA}-J_{\mcB}$ satisfying
\begin{align}
    {\rm tr}J_{\mcA}=\frac{1-\sqrt{p(U_0)}}{1+\sqrt{p(U_0)}}
    J_{\mcB}, \quad J_{\mcA}, J_{\mcB}\geq 0,
\end{align}
shrinks by increasing $p(U_0)$. This matches an intuition that the larger the success probability gets the harder the implementation becomes. Additionally, Eq.~(\ref{eq::sdp_for_prob}) reduces to the SDP for the deterministic and exact case in Thm.~\ref{th::main_sdp} in the limit $p(U_0)\to 1$.

\textbf{Proof of Thm.~\ref{le::prob_exact}:}~
The general probabilistic and exact algorithm to implement $f(U)$ by a fixed-order quantum circuit is represented as the quantum circuit in Fig.~\ref{fig::prob_lower}, where $\mcH$ is the Hilbert space of the input unitary $U\in \mcL(\mcH)$ and the output unitary $f(U)\in \mcL(\mcH)$, $\mcH_A$ is an auxiliary space, and $\rho_A\in \mcH_A$ is a pure state.
\begin{figure}[H]
    \centering
    \includegraphics[width=\linewidth]{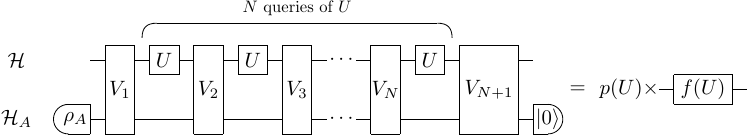}
    \caption{Quantum circuit probabilistically implementing $f(U)$ using $N$ {queries} of black-box unitary operation $U$. Notations follows Fig.~\ref{fig::lowerbound_comb}. $p(U)$ refers to the success probability (probability of measuring $\ket{0}$ in $\mcH_A$) when the input is $p(U)$.}
    \label{fig::prob_lower}
\end{figure}

We first prove the following lemma.
\begin{lem}\label{le::subprob}
    Suppose $N$ is the solution of the SDP~\eqref{eq::sdp_for_prob} at a unitary operation $U_0\in \SU(d)$ for differentiable functions $f:\SU(d)\to \SU(d)$ and $p:\SU(d)\to \RR_{\geq 0}$.
    Then, there does not exist a quantum circuit shown in Fig.~\ref{fig::prob_lower} with less than $N$ queries to $U$ such that
    \begin{enumerate}
        \item[(a)] the probability in which $\ket{0}\in \mcH_A$ is measured for the unitary operation $U$ is exactly equal to $p(U)$, and
        \item[(b)] the output state $\ket{\gamma (U)}\in \mcH$ for the input state $\ket{\psi}\in \mcH$ when $\ket{0}\in \mcH_A$ is measured satisfies $\ket{\gamma (U)} = e^{i\theta (U)}f(U)\ket{\psi}$ for a global phase $\theta (U)$,
    \end{enumerate}
    both in a neighborhood of $U_0$.
\end{lem}

\textbf{Proof:}~
Without loss of generality, we assume $\rho_A=\ketbra{0}{0}$. Defining $\tilde{V}_{N+1}(U_0)\coloneqq(f(U_0)^{-1}\otimes I)V_{N+1}$, we obtain an equation analogous to Eq.~(\ref{eq::deepest_preprocessing}), namely, we have
\begin{align}\label{eq::pro_phase}
    &(I\otimes \bra{0})\left[\tilde{V}_{N+1}(U_0)\left(\prod_{j=1}^{N}(U\otimes I)V_j\right)\right](I\otimes \ket{0})
    \nonumber\\
    =&
    e^{i\theta (U)}\sqrt{p(U)}f(U_0)^{-1}f(U),
\end{align}
which holds in a neighbor of $U_0$.
By setting $U=U_0$, we obtain
\begin{align}
    &(I\otimes \bra{0})\left[\tilde{V}_{N+1}(U_0)\left(\prod_{j=1}^{N}(U_0\otimes I)V_j\right)\right](I\otimes \ket{0})
    \nonumber\\
    =&
    e^{i\theta (U_0)}\sqrt{p(U_0)}I,
\end{align}
thus, using $M^{(s,{\rm right})}_{j,k}(U_0),\ M^{(s,{\rm right})}_{j,k}(U_0)$ defined as
\begin{align}
\begin{split}
    &V^{(s,{\rm left})}(U_0)\eqqcolon\sum_{j,k}
    M^{(s,{\rm left})}_{j,k}(U_0)\otimes \ketbra{j}{k}\\
    &V^{(s,{\rm right})}(U_0)\eqqcolon\sum_{j,k}
    M^{(s,{\rm right})}_{j,k}(U_0)\otimes \ketbra{j}{k}
\end{split}
\end{align}
for $V^{(s,{\rm left})}(U_0),\ V^{(s,{\rm right})}(U_0)$ defined in the same way as in Lem.~\ref{le::Vlefrig}, for all $s\in \{1,\ldots,N\}$, we obtain
\begin{align}
    \sum_j M^{(s,{\rm right})}_{0,j}(U_0)
    M^{(s,{\rm left})}_{j,0}(U_0) =e^{i\theta (U_0)}\sqrt{p(U_0)}I,
\end{align}
and
\begin{align}
    \left(\{M^{(s,{\rm right})}_{0,\ell}(U_0)^{\dagger}\}_{\ell}, \{M^{(s,{\rm left})}_{j,0}(U_0)\}_j\right)
    = e^{i\theta(U_0)}\sqrt{p(U_0)}d.
\end{align}
On the other hand, from the unitarity of $V^{(s,{\rm left})}(U_0),\ V^{(s,{\rm right})}(U_0)$, we obtain
\begin{align}
\begin{split}
    &\left(\{M^{(s,{\rm left})}_{j,0}(U_0)\}_{j}, \{M^{(s,{\rm left})}_{j,0}(U_0)\}_j\right)
    = d,\\
    &\left(\{M^{(s,{\rm right})}_{0,\ell}(U_0)^{\dagger}\}_{\ell}, \{M^{(s,{\rm right})}_{0,\ell}(U_0)^{\dagger}\}_{\ell}\right)
    = d.
\end{split}
\end{align}
Thus, for $A_j^{(s)}(U_0)$ and $B_j^{(s)}(U_0)$ defined as
\begin{align}
\begin{split}
    A_j^{(s)}(U_0)&\coloneqq\frac{1}{2}(M^{(s,{\rm left})}_{j,0}(U_0)+e^{i\theta (U_0)}M^{(s,{\rm right})}_{0,j}(U_0)^{\dagger})\\
    B_j^{(s)}(U_0)&\coloneqq\frac{1}{2}(M^{(s,{\rm left})}_{j,0}(U_0)-e^{i\theta (U_0)}M^{(s,{\rm right})}_{0,j}(U_0)^{\dagger}),
\end{split}
\end{align}
the following equations hold for all $s\in \{1,\ldots ,N\}$:
\begin{align}
\begin{split}
    &{\rm tr}\sum_jA_j^{(s)}(U_0)^{\dagger}A_j^{(s)}(U_0)=\frac{d}{2}(1+\sqrt{p(U_0)}),\\
    &{\rm tr}\sum_jB_j^{(s)}(U_0)^{\dagger}B_j^{(s)}(U_0)=\frac{d}{2}(1-\sqrt{p(U_0)}).
\end{split}
\end{align}

Also, by substituting $U\gets e^{i\epsilon H}U_0$ to Eq.~(\ref{eq::pro_phase}) and differentiating by $\epsilon$ around $\epsilon=0$, we obtain
\begin{align}
    &\left.\frac{{\rm d}}{{\rm d}\epsilon}\right|_{\epsilon=0}
    (I\otimes \bra{0})
    \tilde{V}_{N+1}(U_0)
    \left(\prod_{j=1}^{N}(e^{i\epsilon H}U_0\otimes I)V_j\right)[I\otimes \ket{0}]
    \nonumber\\
    =
    &\sum_{s,j,\ell}\left(M^{(s,{\rm right})}_{0,\ell}(U_0)\otimes \bra{\ell}\right)(iH\otimes I)\left(M^{(s,{\rm left})}_{j,0}(U_0)\otimes \ket{j}\right)
    \nonumber\\
    =&
    \sum_{s,j} 
    M^{(s,{\rm right})}_{0,j}(U_0) iH
    M^{(s,{\rm left})}_{j,0}(U_0)
    \nonumber\\
    =&
    \sum_{s,j}
    e^{i\theta (U_0)}(A_j^{(s)}(U_0)-B_j^{(s)}(U_0))^{\dagger} iH
    (A_j^{(s)}(U_0)+B_j^{(s)}(U_0))
    \nonumber\\
    =&
    \left(\left.\frac{{\rm d}}{{\rm d}\epsilon}\right|_{\epsilon=0}e^{i\theta(e^{i\epsilon H}U_0)}\sqrt{p(e^{i\epsilon H}U_0)}\right)I
    \nonumber\\
    +&ie^{i\theta (U_0)}\sqrt{p(U_0)}g_{U_0}(H),
\end{align}
thus we have
\begin{align}\label{eq::prob_bef}
    &\sum_{s=1}^N\sum_j 
    (A_j^{(s)}(U_0)-B_j^{(s)}(U_0))^{\dagger}H
    (A_j^{(s)}(U_0)+B_j^{(s)}(U_0))
    \nonumber\\
    =&
    -ie^{-i\theta (U_0)}
    \left(\left.\frac{{\rm d}}{{\rm d}\epsilon}\right|_{\epsilon=0}e^{i\theta(e^{i\epsilon H}U_0)}\sqrt{p(e^{i\epsilon H}U_0)}\right)I
    \nonumber\\
    +&\sqrt{p(U_0)}g_{U_0}(H).
\end{align}
By taking the Hermitian part of Eq.~(\ref{eq::prob_bef}), we have
\begin{align}
    &\sum_{s=1}^N\sum_j(A_j^{(s)}(U_0)^{\dagger}HA_j^{(s)}(U_0))
    \nonumber\\
    -&
    \sum_{s=1}^N\sum_j(B_j^{(s)}(U_0)^{\dagger}HB_j^{(s)}(U_0))
    \nonumber\\
    \eqqcolon&\mcA_{U_0}(H)-\mcB_{U_0}(H)
    \nonumber\\
    =&\sqrt{p(U_0)}\alpha_{U_0}(H)I+\sqrt{p(U_0)}g_{U_0}(H),
\end{align}
where $\alpha_{U_0}:\su (d)\to \mathbb{R}$ is a linear map. For these $\mcA_{U_0}$ and $\mcB_{U_0}$, we also have
\begin{align}
    &\left(\sum_{s,j} e^{-i\theta(U_0)}
    M^{(s,{\rm right})}_{0,j}(U_0)
    M^{(s,{\rm left})}_{j,0}(U_0)\right)+
    \nonumber\\
    &\left(\sum_{s,j} e^{-i\theta(U_0)}
    M^{(s,{\rm right})}_{0,j}(U_0)
    M^{(s,{\rm left})}_{j,0}(U_0)\right)^{\dagger}
    \nonumber\\
    = &2\sqrt{p(U_0)}NI
    \nonumber\\
    =&
    2(\mcA_{U_0}(I)-\mcB_{U_0}(I)).
\end{align}
Thus, for $J_{\mcA}, J_{\mcB}, J_{g_{U_0}}$ defined by
\begin{align}
\begin{split}
    &J_{\mcA}\coloneqq\frac{1}{\sqrt{p(U_0)}}(\mcI\otimes \mcA)(\dketbra{I}),\\
    &J_{\mcB}\coloneqq\frac{1}{\sqrt{p(U_0)}}(\mcI\otimes \mcB)(\dketbra{I}),\\
    &J_{g_{U_0}} \coloneqq \sum_{j=1}^{d^2-1} G_j^* \otimes g_{U_0}(G_j),
\end{split}
\end{align}
and $\beta_{U_0}$ satisfying
\begin{align}
\begin{split}
    &\tr(\beta_{U_0}^T I)=N,\\
    &\tr(\beta_{U_0}^T H)=\alpha_{U_0}(H)\quad (H\in \su (d)),
\end{split}
\end{align}
we have
\begin{align}
\begin{split}
    &J_{\mcA}-J_{\mcB}=J_{g_{U_0}}+\beta_{U_0}\otimes I,\\
    &{\rm tr}J_{\mcB}=\frac{1-\sqrt{p(U_0)}}{1+\sqrt{p(U_0)}}
    {\rm tr}J_{\mcA},\\
    &J_{\mcA}, J_{\mcB}\geq 0,
\end{split}
\end{align}
which proves Lem.~\ref{le::subprob}.
\qed

Now we move to the proof of Thm.~\ref{le::prob_exact}. 
Suppose for a sake of contradiction, that for a given set of $N$, $f:U\mapsto f(U)$, and $p:U\mapsto p(U)$, the solution of the SDP in Eq.~(\ref{eq::sdp_for_prob}) is larger than $N$ at some $U_0$, but $f(U)$ can still be implemented by a probability above $p(U)$ with $N$ queries to $U$ with a state $\rho_A$. Then, all eigenvectors of $\rho_A$ with the same sets of $V_1,\ldots ,V_N$ reproduce $f(U)$ exactly with certain probabilities, and in particular, there exists one of its eigenvectors which gives a success probability $p'(U)$ larger than $p(U)$ in a neighborhood of $U_0$.
On the other hand, according to Lem.~\ref{le::subprob}, $N$ has to be larger than the solution of Eq.~(\ref{eq::sdp_for_prob}) for $p'(U)$, which is larger than or equal to that for $p(U)$, leading to a contradiction. 
\qed

As an application of Thm.~\ref{le::prob_exact}, we show the following theorem:
\begin{Theorem}
    When unitary transposition is exactly implemented using $N$ queries to a black-box unitary operation $U$ with probability $p_{\rm trans}(U)$ in a neighborhood of $U_0$ for a differentiable function $p_{\rm trans}$, $p_{\rm trans}(U_0)$ is upper-bounded as
    \begin{align}
        p_{\rm trans}(U_0)\leq \left(\frac{d}{((d^2-1)/N)+1}\right)^2.
    \end{align} 
\end{Theorem}

As a corollary of this theorem, the success probability of transposition with $N=1$ query is shown to be bounded above by $1/d^2$, which is tight since the gate-teleportation-based method shown in \cite{quintino2019probabilistic} achieves this success probability. 

\textbf{Proof:}~  
According to Eq.~(\ref{eq::sdp_for_prob}), there exists $\beta_{U_0}$, $J_{\mcA}$, and $J_{\mcB}$ such that
\begin{align}\label{eq::sdp_hasto}
\begin{split}
    &\tr \beta_{U_0} = N,\\
    & 
    J_{\mcA}-J_{\mcB}=
    \SWAP -\frac{1}{d}I\otimes I + \beta_{U_0} \otimes I,\\
    &{\rm tr}J_{\mcB}=\frac{1-\sqrt{p_{\rm trans}(U_0)}}{1+\sqrt{p_{\rm trans}(U_0)}}
    {\rm tr}J_{\mcA},\\
    &J_{\mcA}, J_{\mcB}\geq 0.
\end{split}
\end{align}
This is because that the solution $N_{\rm min}$ of the SDP in Eq.~(\ref{eq::sdp_for_prob}) has to be smaller than $N$, and defining $\beta'$ as a $\beta_{U_0}$ which gives the solution $N_{\rm min}$, there exist some $J_{\mcA}$ and $J_{\mcB}$ which satisfies Eq.~(\ref{eq::sdp_hasto}) for $\beta_{U_0}\coloneqq \beta'+((N-N_{\rm min})/d)I$.
By sandwiching the second equation of Eq.~(\ref{eq::sdp_hasto}) by $(V\otimes V)$ and $(V\otimes V)^{\dagger}$ and taking the Haar integral over $V$, we have
\begin{align}\label{eq::tr_haared}
\begin{split}
    &J_{\mcA}'-J_{\mcB}'=\SWAP -\frac{1}{d}I\otimes I +\frac{N}{d}I\otimes I,\\
    &{\rm tr}J_{\mcB}'=\frac{1-\sqrt{p_{\rm trans}(U_0)}}{1+\sqrt{p_{\rm trans}(U_0)}}
    {\rm tr}J_{\mcA}',\\
    &J_{\mcA}', J_{\mcB}'\geq 0,
\end{split}
\end{align}
for $J_{\mcA}'$ and $J_{\mcB}'$ defined by
\begin{align}
\begin{split}
    J_{\mcA}'&\coloneqq\int {\rm d}V\ (V\otimes V)J_{\mcA}(V\otimes V)^{\dagger},\\
    J_{\mcB}'&\coloneqq\int {\rm d}V\ (V\otimes V)J_{\mcB}(V\otimes V)^{\dagger}.
\end{split}
\end{align}
Under the first and the third condition of Eq.~(\ref{eq::tr_haared}), the minimum value of ${\rm tr}J_{\mcB}'/{\rm tr}J_{\mcA}'$ is given by $|\sum_{k}\chi_k|/|\sum_{j}\lambda_j|$, where $\SWAP -((N-1)/d)I\otimes I$ is diagonalized as $\sum_j \lambda_j \ketbra{\phi_j}{\phi_j}+\sum_k\chi_k\ketbra{\psi_k}{\psi_k}$ ($\lambda_j\geq 0$, $\chi_k\leq 0$). Therefore, we have
\begin{align}
    \frac{1-\sqrt{p_{\rm trans}(U_0)}}{1+\sqrt{p_{\rm trans}(U_0)}}&\geq 
    \frac{|\sum_k \chi_k|}{|\sum_j \lambda_j|}
    \nonumber\\
    &=
    \frac{{\rm tr}[(1-(N-1)/d)\Pi_{\rm antisym}]}{{\rm tr}[(1+(N-1)/d)\Pi_{\rm sym}]}
    \nonumber\\
    &=\frac{d^2+N-1-Nd}{d^2+N-1+Nd},
\end{align}
thus we obtain
\begin{align}
    p_{\rm trans}(U_0)\leq \left(\frac{d}{((d^2-1)/N)+1}\right)^2 .
\end{align}

\qed

The SDP in Eq.~(\ref{eq::sdp_for_prob}) for unitary inversion and unitary complex conjugation can also be analytically solved in a similar approach.
This shows the upper bounds on the success probabilities $p_\mathrm{conj}(U_0)$ and $p_\mathrm{inv}(U_0)$ of unitary inversion and unitary complex conjugation in a neighborhood of $U_0$, respectively, given by
\begin{align}
\label{eq:p_inv_bound}
    p_\mathrm{inv}(U_0) &\leq \left({d^2\over ((2d^2-2)/N)+d^2-2}\right)^2,\\
\label{eq:p_conj_bound}
    p_\mathrm{conj}(U_0) &\leq \left({d\over ((d^2-1)/N)-1}\right)^2.
\end{align}
The upper bounds may look looser than the previously known no-go results shown as follows:
\begin{itemize}
    \item The fidelity $F$ of approximate unitary inversion is upper bounded by~\cite{chen2025tight}
    \begin{align}
        F\leq {N+1\over d^2},
    \end{align}
    which also shows the upper bound on the success probability $p$ of exact unitary inversion since $p\leq F$ holds.
    \item Unitary complex conjugation cannot have success probability $p>0$ for $N< d-1$~\cite{quintino2019probabilistic}.
\end{itemize}
However, the upper bounds~\eqref{eq:p_inv_bound} and \eqref{eq:p_conj_bound} are also applicable to the case where the success probability is allowed to depend on the input unitary $U$, which is not captured by the previous results.

The SDP in Eq.~(\ref{eq::sdp_for_prob}) cannot be solved analytically in general. Nevertheless, a canonical upper bound of the solution of the SDP is given in the following theorem.
\begin{Theorem}
    Suppose $N$ is the solution of the SDP~\eqref{eq::sdp_for_prob}.
    Then, the success probability $p(U_0)$ is upper bounded as
    \begin{align}
        p(U_0)\leq \left(
        \frac{Nd\|J_{g_{U_0}}\|_{\rm op}}{\|J_{g_{U_0}}\|_2^2}
        \right)^2,
    \end{align}
    where $J_{g_{U_0}}$ is defined as
    \begin{align}
        J_{g_{U_0}}=\sum_{j=1}^{d^1-1}G_j^*\otimes g_{U_0}(G_j)
    \end{align}
    for an orthonormal basis $\{G_j\}_j$ of $\su (d)$ and the linear map $g_{U_0}:\mcL(\mcH)\to\mcL(\mcH)$ is defined by the first-order differentiation of $f$ around $U=U_0$ as
    \begin{align}
        g_{U_0} (H)\coloneqq
        -i\left.\frac{{\rm d}}{{\rm d}\epsilon}\right|_{\epsilon =0}\left[f(U_0)^{-1}f(e^{i\epsilon H}U_0)\right] .
    \end{align}
\end{Theorem}
\textbf{Proof:}~
For a fixed value $N$ of ${\rm tr}\beta_{U_0}$, the highest possible value of $p(U_0)$ such that there exists a set of $\beta_{U_0}$, $J_{\mcA}$, and $J_{\mcB}$ satisfying conditions of SDP in Eq.~(\ref{eq::sdp_for_prob}) satisfies
\begin{align}
    \frac{1-\sqrt{p({U_0})}}{1+\sqrt{p({U_0})}}=r_{\rm min},
\end{align}
namely, we have
\begin{align}
\label{eq:p_U_0}
    p(U_0) = \left(
        \frac{1-r_{\rm min}}{1+r_{\rm min}}
    \right)^2,
\end{align}
where $r_{\rm min}$ is defined as
\begin{align}
    r_{\rm min} \coloneqq \min_{\beta_{U_0}; {\rm tr}\beta_{U_0} =N}\frac{\sum_k(-\chi_k)}{\sum_{j}\lambda_j},
\end{align}
and $\{\lambda_j\}$ and $\{\chi_k\}$ are sets of positive and negative eigenvalues of $J_{g_{U_0}}+\beta_{U_0}\otimes I$, respectively. Since $\sum_j\lambda_j -\sum_k (-\chi_k)={\rm tr}(J_{g_{U_0}}+\beta_{U_0}\otimes I)=Nd$ and $\sum_j\lambda_j +\sum_k (-\chi_k)=\|J_{g_{U_0}}+\beta_{U_0}\otimes I\|_1$ hold, $r_{\rm min}$ can be rewritten as
\begin{align}
\label{eq:r_min}
    r_{\rm min}&=\frac{a-Nd}{a+Nd},
\end{align}
where $a$ is defined by
\begin{align}
    a&\coloneqq\min_{\beta_{U_0}; {\rm tr}\beta_{U_0}=N}
    \|J_{g_{U_0}}+\beta_{U_0}\otimes I\|_1 .
\end{align}
Here, using the inequality $\|AB\|_1\leq \|A\|_1\|B\|_{\rm op}$ and $\|A\|_1\geq |{\rm tr}A|$ for Hermitian $A$ and $B$, we have, for all $\beta_{U_0}$,
\begin{align}
    \|J_{g_{U_0}}+\beta_{U_0}\otimes I\|_1 &=\frac{\|J_{g_{U_0}}\|_{\rm op}\|J_{g_{U_0}}+\beta_{U_0}\otimes I\|_1}{\|J_{g_{U_0}}\|_{\rm op}}
    \nonumber\\
    &\geq 
    \frac{\|J_{g_{U_0}}(J_{g_{U_0}}+\beta_{U_0}\otimes I)\|_1}{\|J_{g_{U_0}}\|_{\rm op}}
    \nonumber\\
    &\geq
    \frac{|{\rm tr}[J_{g_{U_0}}(J_{g_{U_0}}+\beta_{U_0}\otimes I)]|}{\|J_{g_{U_0}}\|_{\rm op}}
    \nonumber\\
    &=\frac{\|J_{g_{U_0}}\|_2^2}{\|J_{g_{U_0}}\|_{\rm op}}.
\end{align}
Therefore, we obtain
\begin{align}
    a\geq \frac{\|J_{g_{U_0}}\|_2^2}{\|J_{g_{U_0}}\|_{\rm op}}.
\end{align}
By substituting this inequality back to Eqs.~\eqref{eq:p_U_0} and \eqref{eq:r_min}, we have
\begin{align}
    p(U_0)\leq \left(
    \frac{Nd\|J_{g_{U_0}}\|_{\rm op}}{\|J_{g_{U_0}}\|_2^2}
    \right)^2.
\end{align}
\qed

\section{Derivation of the dual SDPs}\label{app::duals}
In this section, we derive the dual problems for the SDPs \eqref{eq:sdp} in the main text, \eqref{eq::sdp_subgroup}, and \eqref{eq::sdp_for_prob} shown in this Supplemental Material.

The SDP~\eqref{eq:sdp} is given by the following optimization problem:
\begin{align}
    \min_{\beta_{U_0}} \max_{\Gamma\geq 0} \mathcal{L},
\end{align}
where $\mathcal{L}$ is the Lagrangian defined by
\begin{align}
    \mathcal{L} &\coloneqq \Tr \beta_{U_0} - \Tr[(J_{g_{U_0}}+\beta_{U_0} \otimes I)\Gamma]\\
    &=\Tr[(I-\Tr_{2}\Gamma)\beta_{U_0}] - \Tr(J_{g_{U_0}}\Gamma),
\end{align}
by introducing a dual variable $\Gamma \in \mcL(\CC^d \otimes \CC^d)$.
The dual problem is obtained by considering the following optimization problem:
\begin{align}
    \max_{\Gamma\geq 0} \min_{\beta_{U_0}}  \mathcal{L},
\end{align}
which reduces to the dual problem given by
\begin{align}
\begin{split}
    \max &- \Tr(J_{g_{U_0}}\Gamma)\\
    \text{s.t. }& \Gamma\geq 0,\\
    &\Tr_{2}\Gamma = I.
\end{split}
\end{align}

The SDP \eqref{eq::sdp_subgroup} is given by the following optimization problem:
\begin{align}
    \min_{\beta_{U_0}, \{B'_k\}} \max_{\Gamma\geq 0, \{\lambda_k\}} \mathcal{L},
\end{align}
where $\mathcal{L}$ is the Lagrangian given by
\begin{align}
    \mathcal{L}\coloneqq & \Tr\beta_{U_0} + \sum_k \lambda_k \Tr(B'_k)\nonumber\\
    &- \Tr[(\sum_j G_j^* \otimes g_{U_0}(G_j) + \sum_k B_k^*\otimes B'_k + \beta_{U_0}\otimes I)\Gamma] \\
    =& -\Tr[(\sum_j G_j^* \otimes g_{U_0}(G_j))\Gamma]\nonumber\\
    &-\sum_k \Tr[B'_k [\Tr_1((B_k^* \otimes I)\Gamma)-\lambda_k I]]\nonumber\\
    &-\Tr[\beta_{U_0} (I-\Tr_2 \Gamma)],
\end{align}
by introducing dual variables $\Gamma \in \mcL(\CC^d\otimes \CC^d)$ and $\lambda_k\in \RR$.
The dual problem is obtained by considering the following optimization problem:
\begin{align}
    \max_{\Gamma\geq 0, \{\lambda_k\}} \min_{\beta_{U_0}, \{B'_k\}} \mathcal{L},
\end{align}
which reduces to the dual problem given by
\begin{align}
\begin{split}
    &\max -\Tr[(\sum_j G_j^* \otimes g_{U_0}(G_j))\Gamma]\\
    \mathrm{s.t.}\;& \Gamma\geq 0, \lambda_k\in\RR,\\
    &\Tr_1[(B_k^*\otimes I) \Gamma] = \lambda_k I \quad \forall k,\\
    &\Tr_2\Gamma = I.
\end{split}
\end{align}
The dual variable $\lambda_k$ can be removed since the dual SDP constraints imply
\begin{align}
    d \lambda_k
    &= \Tr[(B_k^* \otimes I)\Gamma]\\
    &= \Tr[B_k^*\Tr_2(\Gamma)]\\
    &= \Tr(B_k^*)\\
    &= 0.
\end{align}
Thus, we obtain
\begin{align}
\begin{split}
    &\max -\Tr[(\sum_j G_j^* \otimes g_{U_0}(G_j))\Gamma]\\
    \mathrm{s.t.}\;& \Gamma\geq 0,\\
    &\Tr_1[(B_k^*\otimes I) \Gamma] = 0 \quad \forall k,\\
    &\Tr_2\Gamma = I.
\end{split}
\end{align}

The SDP \eqref{eq::sdp_for_prob} is given by the following optimization problem:
\begin{align}
    \min_{J_{\mcA},J_{\mcB}\geq 0, \beta_{U_0}\in \mcL(\mathbb{C}^d)} \max_{M\in \mcL(\mathbb{C}^{d}\otimes \mathbb{C}^{d}), a\in \mathbb{R}} \mathcal{L},
\end{align}
where $\mcL$ is the Lagrangian given by
\begin{align}
    \mcL\coloneqq&{\rm tr}\beta_{U_0}+{\rm tr}[M(J_{\mcA}-J_{\mcB}-J_{g_{U_0}}-\beta_{U_0}\otimes I)]
    \nonumber\\
    &+a{\rm tr}\left(
    J_{\mcB}-\frac{1-\sqrt{p(U_0)}}{1+\sqrt{p(U_0)}}J_{\mcA}
    \right)
    \nonumber\\
    =&-{\rm tr}(MJ_{g_{U_0}})+{\rm tr}\left[J_{\mcA}\left(M-a\frac{1-\sqrt{p(U_0)}}{1+\sqrt{p(U_0)}}I\right)\right]\nonumber\\
    &+{\rm tr}[J_{\mcB}(aI-M)]+{\rm tr}_1[\beta_{U_0}(I-{\rm tr}_2M)],
\end{align}
by introducing dual variables $M\in \mcL(\mathbb{C}^d\otimes \mathbb{C}^d)$ and $a\in \mathbb{R}$. 
The dual problem is obtained by considering the following optimization problem:
\begin{align}
    \max_{M\in \mcL(\mathbb{C}^{d}\otimes \mathbb{C}^{d}), a\in \mathbb{R}}
    \min_{J_{\mcA},J_{\mcB}\geq 0, \beta_{U_0}\in \mcL(\mathbb{C}^d)} 
    \mathcal{L},
\end{align}
which reduces to the dual problem given by
\begin{align}
\begin{split}
    \max &- \Tr(J_{g_{U_0}}M)\\
    \text{s.t. }& M\in \mcL(\mathbb{C}^d\otimes\mathbb{C}^d),\\
    & a\in \mathbb{R},\\
    & M-a\frac{1-\sqrt{p(U_0)}}{1+\sqrt{p(U_0)}}I\geq 0,\\
    & aI-M\geq 0,\\
    &\Tr_{2}M = I.
\end{split}
\end{align}

\end{document}